\documentclass[12pt,a4paper,reqno]{amsart}
\usepackage{t1enc}
\usepackage{times}
\usepackage{amssymb}
\usepackage[mathscr]{euscript}
\usepackage{graphicx}

\usepackage[final]{changes}
\colorlet{Changes@Color}{red}

\newtheorem{thm}{Theorem}[section]
\newtheorem{prop}[thm]{Proposition}
\newtheorem{lem}[thm]{Lemma}
\newtheorem{cor}[thm]{Corollary}

\theoremstyle{remark}

\theoremstyle{definition}

\renewcommand{\phi}{\varphi} %nice phi?
\newcommand{\tr}{\mathrm{Tr}} %Trace
\newcommand{\C}{\mathcal C} %Hilbert Schmidt operator
\newcommand{\E}{\mathbb{E}} %Expectation
\newcommand{\T}{\mathcal T} %Linear operator
\newcommand{\U}{\mathcal U} %Semigroup notation
\renewcommand{\H}{\mathcal H}

\newcommand{\cut}{\mathrm{cut}} %cut function

\newcommand{\<}{\langle}
\renewcommand{\>}{\rangle}
\newcommand{\Id}{\mathrm{Id}} %Identity operator
\newcommand{\ran}{\mathrm{ran}} %Range
\newcommand{\dom}{\mathrm{dom}} %Domain
\newcommand{\nooutput}[1]{}
\begin{document}
 
\date{Version of \today}

\title[Finite dimensional approximation of forward prices]{Approximation of forward curve models in commodity markets with arbitrage-free 
finite dimensional models}
\begin{abstract}
  In this paper we show how to approximate a Heath-Jarrow-Morton dynamics for the forward prices in commodity markets with arbitrage-free models which have
a finite dimensional state space. Moreover, we recover a closed form representation of the forward price dynamics in the approximation models and derive the rate of convergence uniformly over an interval of time to maturity to the true dynamics under certain additional smoothness conditions. In the Markovian case we can strengthen the convergence to be uniform over time as well. Our results are based on the construction of a convenient 
Riesz basis on the state space of the term structure dynamics. 
\end{abstract}

\author[Benth]{Fred Espen Benth}
\address[Fred Espen Benth]{\\
Department of Mathematics \\
University of Oslo\\
P.O. Box 1053, Blindern\\
N--0316 Oslo, Norway}
\email[]{fredb\@@math.uio.no}
\urladdr{http://folk.uio.no/fredb/}
\author[Kr\"uhner]{Paul Kr\"uhner}
\address[Paul Kr\"uhner]{\\ 
Financial \& Actuarial Mathematics\\
Vienna University of Technology\\
Wiedner Hauptstr. 8/E105-1\\
AT-1040 Vienna, Austria}
\email[]{paulkrue@fam.tuwien.ac.at}
\urladdr{https://fam.tuwien.ac.at/~paulkrue/}

\thanks{F.\ E.\ Benth acknowledges financial support from the project "Managing Weather Risk in Energy Markets (MAWREM)", funded by the ENERGIX program
of the Norwegian Research Council.}

\subjclass[2010]{91B24, 91G20}

\keywords{Energy markets, Heath-Jarrow-Morton, Non harmonic Fourier analysis, arbitrage free approximations}

\maketitle

\section{Introduction}
We develop arbitrage-free approximations to the forward term structure dynamics in commodity 
markets. The approximating term structure models have finite dimensional state space, and therefore 
tractable for further analysis and numerical simulation. We provide results on the convergence of the 
approximating term structures and characterize the speed under reasonable smoothness properties
of the true term structure. Our results are based on the construction of a convenient Riesz basis on the state space of the term structure dynamics.  

In the context of fixed-income markets, Heath, Jarrow and Morton~\cite{HJM} propose to  model the entire term structure of interest rates. Filipovi\'c \cite{filipovic.01} reinterprets this approach in the so-called Musiela parametrisation, i.e., studying the so-called forward rates as solutions of first-order stochastic 
partial differential equations. This class of stochastic partial differential equations is often referred to as
the Heath-Jarrow-Morton-Musiela (HJMM) dynamics. This highly successful method has been transferred to 
other markets, including commodity and energy futures markets (see Clewlow and 
Strickland~\cite{CS} and Benth,
Saltyte Benth and Koekebakker~\cite{BSBK-book}), where the term structure of forward and 
futures prices are modelled by similar stochastic partial differential equations.

An important stream of research in interest rate modelling has been so-called finite dimensional realizations
of the solutions of the HJMM dynamics (see e.g., Bj\"ork and Svensson~\cite{BjorkSvensson}, Bj\"ork and Landen~\cite{BjorkLanden}, Filipovic and Teichmann~\cite{FT} and Tappe~\cite{Tappe2010}). Starting out with an equation for the forward rates driven by a 
$d$-dimensional Wiener process, the question has been under what conditions on the volatility and drift 
do we get solutions which belongs to a finite dimensional space, that is, when can the dynamics of the
whole curve be decomposed into a finite number of factors. This property has a close connection with
principal component analysis (see Carmona and Tehranchi~\cite{CT}), but is also convenient when 
it comes to further analysis like estimation, simulation, pricing and portfolio management (see 
Benth and Lempa~\cite{BL} for the latter).  

In energy markets like power and gas, there is empirical and economical evidence for high-dimensional 
noise. Moreover, the noise shows clear leptokurtic signs (see Benth, \v{S}altyt\.e Benth and Koekebakker~\cite{BSBK-book} and references therein). These empirical insights motivate the use of 
infinite dimensional L\'evy processes driving the noise in the HJMM-dynamics modelling the forward term structure. We refer to Carmona and Tehranchi~\cite{CT} for a thorough analysis of HJMM-models with
infinite dimensional Gaussian noise in interest rate markets. Benth and Kr\"uhner~\cite{BK-stochastics}
introduced a convenient class of infinite dimensional L\'evy processes via subordination of Gaussian 
processes in infinite dimensions. These models were used in analysing stochastic partial differential 
equations with infinite dimensional L\'evy noise in Benth and Kr\"uhner~\cite{BK-coms}. Further,
pricing and hedging of derivatives in energy markets based on such models were studied in
Benth and Kr\"uhner~\cite{BK-siam}. 

The present paper is motivated by the need of an arbitrage-free approximation of Heath, Jarrow, Morton style models -- using the Musiela parametrisation -- in electricity finance. Related research has been carried out by Henseler, Peters and Seydel \cite{Henseler.al.15} who construct a finite-dimensional affine model where a refined principle component analysis (PCA) method does yield an arbitrage free approximation of the term structure model.
Our main result Theorem~\ref{t:main statement} states that the arbitrage-free models for the underlying forward curve
process $f(t,x)$, $x\geq0$ being time to maturity and $t\geq 0$ is current time, can be approximated with processes of the form
 $$ 
f_k(t,x) = S_k(t) + \sum_{n=-k}^k U_n(t)g_n(x) \,,
$$
where $S_k$ denotes the spot prices in the approximating model, $g_{-k},\dots,g_k$ are deterministic functions and $U_{-k},\dots,U_k$ are one-dimensional Ornstein Uhlenbeck type processes. Obviously, models of this type are much easier to handle in applications than general solutions for the HJMM equation. The 
approximation $f_k$ is again a solution of an HJMM equation, and as such being an arbitrage-free model
for the forward term structure. We prove a uniform convergence in space of $f_k$ to the "real" forward 
price curve $f$, pointwise in time. The convergence rate is of order $k^{-1}$ when the 
forward curve $x\mapsto f(t,x)$ is twice continuously differentiable. Our approach
is an alternative to numerical approximations of the HJMM dynamics based on finite difference schemes or finite element
methods, where arbitrage-freeness of the approximating dynamics is not automatically ensured.
We refer to Barth~\cite{Barth} for an analysis of finite element methods 
applied to
stochastic partial differential equations of the type we study.

We refine our results to the Markovian case, where the convergence is slightly strengthened to be uniform over time as well.  Our approach goes via the explicit construction of
a Riesz basis for a subspace of the so-called Filipovi\'c space (see Filipovi\'c~\cite{filipovic.01}),
a separable Hilbert space of absolutely continuous functions on the positive real line with
(weak) derivative disappearing at a certain speed at infinity. The basis will be the functions $g_n$ in the
approximation $f_k$, and the subspace is defined by concentrating the functions in the Filipovi\'c space
to a finite time horizon $x\leq T$. This space was defined in Benth and Kr\"uhner~\cite{BK-coms}, 
and we extend the analysis here to accomodate the arbitrage-free finite dimensional approximation of 
the HJMM-dynamics. We rest on properties of $C_0$-semigroups and stochastic integration with
respect to infinite dimensional L\'evy processes (see Peszat and Zabczyk~\cite{peszat.zabczyk.07}) 
in the analysis.

This paper is organised as follows. In Section~\ref{s:the model} we start with the mathematical formulation 
of the HJMM dynamics for forward rates set in the Filipovi\'c space. The Riesz basis that will make the foundation for our approximation is defined and analysed in detail in Section~\ref{s:mathematical Preliminaries}. The arbitrage-free finite dimensional approximation to term structure modelling is 
constructed in Section~\ref{s:approximation}, where we study convergence properties. The Markovian
case is analysed in the last Section~\ref{s:markovian}.

\section{The model of the forward price dynamics}
\label{s:the model}

Throughout this paper we use the Hilbert space
$$
H_\alpha := \left\{f\in AC(\mathbb R_+,\mathbb C): \int_0^\infty |f'(x)|^2e^{\alpha x} dx <\infty\right\}\,,
$$
where $AC(\mathbb R_+,\mathbb C)$ denotes the space of complex-valued absolutely continuous functions on $\mathbb R_+$. We endow
$H_{\alpha}$ with the scalar product $\<f,g\>_\alpha:=f(0)\overline{g}(0) + \int_0^\infty f'(x)\overline{g}'(x) e^{\alpha x}dx$, and denote
the associated norm by $\|\cdot\|_{\alpha}$. Filipovi\'c~\cite[Section 5]{filipovic.01} shows that $(H_\alpha,\|\cdot\|_\alpha)$ is a separable Hilbert space\footnote{Note that
Filipovi\'c~\cite{filipovic.01} does not consider complex-valued functions. In our context, this minor extension is convenient, as will be clear later.}. 
This space has been used in Filipovi\'c~\cite{filipovic.01} for term structure modelling of bonds and many mathematical properties have been derived therein.
We will frequently refer to $H_{\alpha}$ as the {\it Filipovi\'c space}.

We next introduce our dynamics for the term structure of forward prices in a commodity market. Denote by $f(t,x)$ the price at time $t$ of
a forward contract where time to delivery of the underlying commodity is $x\geq 0$. We treat $f$ as a stochastic process in time with values in
the Filipovi\'c space $H_{\alpha}$. More specifically, we assume that the process $\{f(t)\}_{t\geq 0}$ follows the HJM-Musiela model
which we formalize next.

On a complete filtered probability space $(\Omega,\{\mathcal{F}_t\}_{t\geq 0},\mathcal{F},P)$, where the filtration is assumed to be complete and right continuous, 
we work with an $H_{\alpha}$-valued L\'evy process $\{L(t)\}_{t\geq 0}$ (cf.\ Peszat and Zabczyk~\cite[Theorem 4.27(i)]{peszat.zabczyk.07} for the construction of $H_\alpha$-valued L\'evy processes). We assume that $L$ has finite variance and mean equal to zero, and denote its covariance operator by $\mathcal Q$. 
Let $f_0\in H_\alpha$ and $f$ be the solution of the stochastic partial differential equation (SPDE)
\begin{equation}
\label{e:HJMM-equation}
 df(t) = \partial_x f(t) dt + \beta(t) dt + \Psi(t)dL(t),\quad t\geq 0, f(0)=f_0
\end{equation}
where $\beta\in L^1((\Omega\times\mathbb R_+,\mathcal P,P\otimes\lambda),H_\alpha)$, $\mathcal P$ being the predictable $\sigma$-field, and 
$\Psi \in \mathcal{L}^2_{L}(H_\alpha):=\bigcup_{T>0}\mathcal{L}^2_{L,T}(H_\alpha) $ where the latter space is defined as in {Peszat and Zabczyk~\cite[page 113]{peszat.zabczyk.07}}. For $t\geq 0$, denote by $\mathcal{U}_t$ the shift semigroup on $H_{\alpha}$ defined by
$\mathcal{U}_t f=f(t+\cdot)$ for $f\in\H_{\alpha}$. It is shown in Filipovi\'c~\cite{filipovic.01} that $\{\mathcal{U}_t\}_{t\geq 0}$ is a $C_0$-semigroup on 
$H_{\alpha}$, with generator $\partial_x$. Recall, that any $C_0$-semigroup admits the bound $\Vert\mathcal U_t\Vert_{\mathrm{op}}\leq Me^{wt}$ for some $w,M>0$ and any $t\geq 0$. Here, $\|\cdot\|_{\mathrm{op}}$
denotes the operator norm. In fact, in Filipovi\'c~\cite[Equation (5.10)]{filipovic.01} and Benth and Kr\"uhner~\cite[Lemma~3.4]{benth.kruehner.15} 
it is shown that $\Vert\mathcal U_t\Vert_{\mathrm{op}}\leq C_{\mathcal{U}}$ for any $t\geq 0$ and a constant $C_{\mathcal{U}}:=\sqrt{2(1\wedge\alpha^{-1})}$. Thus $s\mapsto \mathcal U_{t-s}\beta(s)$ is Bochner-integrable and $s\mapsto \mathcal U_{t-s}\Psi(s)$ is integrable with respect to $L$. The unique mild solution of \eqref{e:HJMM-equation} is
\begin{equation}
\label{e:HJMM-equation-mild}
f(t)=\mathcal{U}_tf_0+\int_0^t\mathcal{U}_{t-s}\beta(s)\,ds+\int_0^t\mathcal{U}_{t-s}\Psi(s)\,dL(s)\,. 
\end{equation}

If we model the forward price dynamics $f$ in a risk-neutral setting, the drift coefficient $\beta(t)$ will naturally be zero in order to ensure
the (local) martingale property of the process $t\mapsto f(t,\tau-t)$, where $\tau\geq t$ is the time of delivery of the forward. In this
case, the probability $P$ is to be interpreted as the equivalent martingale measure (also called the pricing measure). However, with
a non-zero drift, the forward model is stated under the market probability and $\beta$ can be related to the risk premium in the market. 

In energy markets like power and gas, the forward contracts deliver over a period, and forward prices can be expressed
by integral operators on the Filipovi\'c space applied on $f$ (see Benth and Kr\"uhner~\cite{benth.kruehner.14,benth.kruehner.15} for more details). 

The dynamics of $f$ can also be considered as a model for the forward rate in fixed-income theory, see Filipovi\'c~\cite{filipovic.01}.  This is indeed the 
traditional application area and point of analysis of the SPDE in \eqref{e:HJMM-equation}. Note, however, that the original no-arbitrage condition in the HJM approach for interest rate markets is different from the no-arbitrage condition used here. If $f$ is understood as the forward rate modelled
in the risk-neutral setting, there is a no-arbitrage relationship between the drift $\beta$, the volatility $\sigma$ and the covariance of the driving noise $L$. We refer to Carmona and Tehranchi~\cite{CT} for a detailed analysis.

\section{A Riesz basis for the Filipovi\'c space}\label{s:mathematical Preliminaries}

In this section we introduce a Riesz basis for a suitable subspace of $H_\alpha$ 
defined in Benth and Kr\"uhner~\cite[Appendix A]{benth.kruehner.14} and present various of its 
properties. Moreover, we give refined statements for this basis and also identify new properties. 
We recall from Young~\cite{Young.80} that any Riesz basis $\{g_n\}_{n\in\mathbb{N}}$ on a separable Hilbert space can 
be expressed by $g_n = \T e_n$ where $\{e_n\}_{n\in\mathbb N}$ is an orthonormal basis and $\T$ is a bounded invertible linear operator. 
For further properties and definitions of Riesz bases, see Young~\cite{Young.80}. 

In Section \ref{s:approximation} we want to employ the spectral method to an approximation of the SPDE in \eqref{e:HJMM-equation} involving the differential operator on the Filipovi\'c space $H_\alpha$. Thus, it would be convenient to have available the eigenvector basis for the differential operator. However, its eigenvectors do not seem to have nice basis properties. Instead, we propose to use a system of vectors which forms a Riesz basis which turns out to be almost an eigenvector system for the differential operator. This property will be made precise in Propositions~\ref{p:U and the riesz basis} and \ref{l:commutator of U and projectors}. 
Finally, we will identify the convergence speed of the Riesz basis expansion. 

Fix $\lambda>0$, $T>0$, and introduce
\begin{equation}
\mathrm{cut}:\mathbb R_+\rightarrow [0,T)\,,\qquad x\mapsto x-\max\{Tz: z\in\mathbb Z:Tz\leq x\}\,,
\end{equation}
and
\begin{equation}
\label{def-A-operator}
   \mathcal A:L^2([0,T),\mathbb C)\rightarrow L^2(\mathbb R_+,\mathbb C)\,,\qquad 
   f\mapsto \left(x\mapsto e^{-\lambda x}f(\mathrm{cut}(x))\right)\,.
  \end{equation}
Here, $L^2(A,\mathbb C)$ is the space of complex-valued square integrable functions on the Borel set $A\subset\mathbb R_+$ equipped with the 
Lebesgue measure. The inner product of $L^2(A,\mathbb C)$ will be denoted $(\cdot,\cdot)_2$ and the corresponding norm $|\cdot|_2$. We remark that 
the set $A$ will be clear from the context and thus not indicated in the notation for
norm and inner product.
 
We define
   \begin{align}
     g_*(x) &:= 1, \label{e:g-star-def}\\
     g_n(x) &:= \frac{1}{\lambda_n\sqrt{T}}\left(\exp\left(\lambda_nx\right)-1\right)\,, \label{e:g-n-def}
   \end{align}
where 
\begin{equation}
\label{e:lambda-n-def}
\lambda_n:=\frac{2\pi i}{T}n-\lambda-\frac{\alpha}{2}\,,
\end{equation} 
for any $n\in\mathbb Z$, $x\geq0$.
It is simple to verify that $g_n\in H_\alpha$ for any $n\in\mathbb Z$ and $g_*\in H_\alpha$. 
As we will see, the system of vectors $\{g_*,\{g_n\}_{n\in\mathbb Z}\}$ forms a Riesz basis and we will use this to obtain arbitrage-free finite-dimensional approximations of the forward price dynamics \eqref{e:HJMM-equation}. 
  
We start our analysis with some elementary properties of the operator $\mathcal A$ which have been proven in Benth and Kr\"uhner~\cite{benth.kruehner.14}.
\begin{lem}\label{l:stetige Einbettung}
  $\mathcal A$ is a bounded linear operator and its range is closed in $L^2(\mathbb R_+,\mathbb C)$. Moreover,
  $$ \frac{e^{-2T\lambda}}{1-e^{-2T\lambda}}| f|_2^2 \leq |\mathcal Af|_2^2 \leq \frac{1}{1-e^{-2T\lambda}}| f|_2^2$$
  for any $f\in L^2([0,T),\mathbb C)$.
\end{lem}
\begin{proof}
 This proof can be found in Benth and Kr\"uhner~\cite[Lemma A.1]{benth.kruehner.14}.
\end{proof}

In the following Proposition~\ref{p:Riesz basis on L2}, we calculate a Riesz basis of the space $\ran(\mathcal{A})$
and its biorthogonal system. The Riesz basis will be given as the image of an orthonormal basis of $L^2([0,T),\mathbb C)$. Consequently, its biorthogonal 
system is given by the image of $(\mathcal A^{-1})^*$, which we calculate in the Lemma below:
\begin{lem}
\label{lem:dual_A}
 The dual $(\mathcal A^{-1})^*$ of the inverse of $\mathcal A:L^2([0,T),\mathbb C)\rightarrow \ran(\mathcal{A})$
is given by
 \begin{align*}
  (\mathcal A^{-1})^*&:L^2([0,T),\mathbb C)\rightarrow \ran(\mathcal{A}) ,\\
           (\mathcal A^{-1})^*f(x)&=(1-e^{-2\lambda T})e^{-\lambda x}\left( e^{2\lambda \cut(x)}f(\cut(x)) \right) \\
             &= (1-e^{-2\lambda T})e^{2\lambda \cut(x)} \mathcal Af(x),\quad x\geq0\,.
 \end{align*}
\end{lem}
\begin{proof}
Let $f,g\in L^2([0,T],\mathbb C)$
and define $h(x):=(1-e^{-2\lambda T})e^{2\lambda \cut(x)} \mathcal Af(x)$ for any $x\geq0$. Then we have
\begin{align*}
  (h,\mathcal Ag)_2&= \int_0^\infty h(y) \overline{\mathcal Ag(y)} dy \\
             &= (1-e^{-2\lambda T}) \sum_{n=0}^\infty \int_{nT}^{(n+1)T} e^{2\lambda(x-nT)} (e^{-\lambda x}f(x-nT))(e^{-\lambda x}\overline{g(x-nT)}) dx \\
             &= (1-e^{-2\lambda T})\sum_{n=0}^\infty e^{-2\lambda nT} \int_{nT}^{(n+1)T} f(x-nT)\overline{g(x-nT)} dx \\
             &= \int_0^T f(y)\overline{g(y)} dy\,.
 \end{align*}
On the other hand,
\begin{align*}
 ((\mathcal A^{-1})^*f, \mathcal Ag)_2&=(f,g)_2= \int_0^Tf(y)\overline{g(y)} dy\,.
 \end{align*}
Since $g$ is arbitrary, we have $h = (\mathcal A^{-1})^*f$ as claimed.
\end{proof}

Parts of the next proposition can be found in Benth and Kr\"uhner~\cite[Lemma A.3]{benth.kruehner.14}. In that paper there appears to be a gap in the proof which we have filled here. 
\begin{prop}\label{p:Riesz basis on L2}
   Define
   $$ e_n(x) := \frac{1}{\sqrt{T}} \exp\left(\left(\frac{2\pi in}{T}-\lambda\right)x\right),\quad x\geq 0,n\in\mathbb Z.$$
   Then $\{e_n\}_{n\in\mathbb Z}$ is a Riesz basis on the closed subspace $\ran(\mathcal{A})$ of $L^2(\mathbb R_+,\mathbb C)$ and
    $$ F:=\{ f\in L^2(\mathbb R_+,\mathbb C): f(x)=0,x\in[0,T) \} $$
   is a closed vector space compliment of $\ran(\mathcal{A})$. The continuous linear projector $\mathcal P_{\mathcal A}$ with range $\ran(\mathcal{A})$ and kernel $F$ has operator norm $\sqrt{\frac{1}{1-e^{-2\lambda T}}}$ and we have
     $$ \mathcal P_{\mathcal A}f(x) = f(x),\quad x\in[0,T], f\in L^2(\mathbb R_+,\mathbb C).$$
     
    The biorthogonal system $\{e_n\}^*_{n\in\mathbb Z}$ for the Riesz basis $\{e_n\}_{n\in\mathbb Z}$ is given by
     $$ e_n^*(x) = \left(1-e^{-2\lambda T}\right)e^{2\lambda\cut(x)} e_n(x) $$
\end{prop}
\begin{proof}
Recall that the range of $\mathcal A$ is a closed subspace of $L^2(\mathbb R_+,\mathbb C)$ due to the lower bound given in Lemma \ref{l:stetige Einbettung}. Furthermore, $\{b_n\}_{n\in\mathbb Z}$ with
  $$ b_n(x):= \frac{1}{\sqrt{T}}\exp\left(\frac{2\pi i n}{T}x\right),\quad n\in\mathbb Z,x\in[0,T)$$
 is an orthonormal basis of $L^2([0,T],\mathbb C)$. Observe, that $e_n = \mathcal Ab_n$ and hence $\{e_n\}_{n\in\mathbb Z}$ becomes a Riesz basis of 
$\ran(\mathcal{A})$.
 
 Define the continuous linear operators
 \begin{align*}
   \mathcal M_{\lambda} &:L^2(\mathbb [0,T),\mathbb C)\rightarrow L^2([0,T),\mathbb C),\mathcal M_{\lambda}f(x):=e^{\lambda x}f(x),\\
   \mathcal C &:L^2(\mathbb R_+,\mathbb C)\rightarrow L^2([0,T),\mathbb C),f\mapsto f\vert_{[0,T)}
 \end{align*}
 and $\mathcal P_{\mathcal A}:=\mathcal A\mathcal M_{\lambda}\mathcal C$. Observe, that $\mathcal M_{\lambda}\mathcal C\mathcal A$ is the identity operator on $L^2([0,T),\mathbb C)$ and hence $\mathcal P_{\mathcal A}^2=\mathcal P_{\mathcal A}$. Therefore, $\mathcal P_{\mathcal A}$ is a continuous linear projection with kernel $F$ and range $\ran(\mathcal{A})$.
  
 Let $f\in L^2(\mathbb R_+,\mathbb C)$ be orthogonal to any element of the kernel of $\mathcal P_{\mathcal A}$. Then $f(x)=0$ Lebesgue-a.e.\ for any $x\geq T$. Hence, we have
\begin{align*}
|\mathcal P_{\mathcal A}f|_2^2 &= \sum_{n\in\mathbb N} \int_{nT}^{nT+T} (e^{-\lambda x} e^{\lambda(x-nT)})^2|f(x-nT)|^2dx \\
                           &= \sum_{n\in\mathbb N} e^{-2n\lambda T} \vert f\vert_2^2 \\
                           &= \frac{1}{1-e^{-2\lambda T}} \vert f\vert_2^2
 \end{align*}
 and it follows that $\Vert\mathcal P_{\mathcal A}\Vert_{\mathrm{op}} = \sqrt{\frac{1}{1-e^{-2\lambda T}}}$.
 
According to Lemma~\ref{lem:dual_A}, we have
 \begin{align*}
   e_n^*(x) &= (\mathcal{A}^{-1})^*b_n(x) \\
            &= (1-e^{-2\lambda T})e^{-\lambda x}\left( e^{2\lambda \cut(x)}b_n(\cut(x)) \right) \\
            &= \left(1-e^{-2\lambda T}\right)e^{2\lambda\cut(x)} e_n(x)\,,
 \end{align*}
 for any $n\in\mathbb Z$, $x\geq0$, as required.
\end{proof}

The statements collected in this section have been about the space $L^2(\mathbb R_+,\mathbb C)$ so far. However, we are actually interested in the space $H_\alpha$ which has a natural and simple isometry to $\mathbb C\times L^2(\mathbb R_+,\mathbb C)$. The next corollary will translate the $L^2(\mathbb R_+,\mathbb C)$-statements above to $H_\alpha$. Before stating it, we introduce a notation for later use:
 Define
\begin{equation}
\label{def:Theta-op}
\Theta:H_\alpha\rightarrow \mathbb C\times L^2(\mathbb R_+,\mathbb C), f\mapsto (f(0), w_\alpha f')\,,
\end{equation} 
where $w_\alpha(x):=e^{x\alpha/2}$ for $x\geq 0$. Then $\Theta$ is an isometry of Hilbert spaces. Its inverse is given by
\begin{equation}
\label{def:Theta-inv-op}
\Theta^{-1}:\mathbb C\times L^2(\mathbb R_+,\mathbb C)\rightarrow H_\alpha,(z,f)\mapsto z+\int_0^{(\cdot)} w_{\alpha}^{-1}(y)f(y)dy\,.
\end{equation}
We use these operators to prove:
\begin{cor}\label{k:Riesz basis on H}
The system $\{g_*,\{g_n\}_{n\in\mathbb Z}\}$ defined in \eqref{e:g-star-def}-\eqref{e:g-n-def} is a Riesz basis of a closed subspace $H_\alpha^{T}$ of $H_\alpha$.
Indeed, $H_\alpha^{T}$ is the space generated by 
$\{g_*,\{g_n\}_{n\in\mathbb Z}\}$.
 Moreover, there is a continuous linear projector $\Pi$ with range $H_\alpha^{T}$ and operator norm $\sqrt{\frac{1}{1-e^{-2\lambda T}}}$ such that
    $$ \Pi h(x) = h(x), \quad h\in H_\alpha,x\in[0,T]. $$
   Consequently, $\Pi\U_th(x) = \U_t\Pi h(x) = h(x+t)$ for any $t\in[0,T]$ and any $x\in[0,T-t]$. 
   
  The biorthogonal system $\{g_*^*,\{g_n^*\}_{n\in\mathbb Z}\}$ is given by
   \begin{align*}
      g_*^*(x) &= 1 \\
      g_n^*(x) &= \int_0^x e^{-y\frac{\alpha}{2}} e_n^*(y) dy
   \end{align*}
where $e_n^*$ is given in Proposition \ref{p:Riesz basis on L2} for any $n\in\mathbb Z$, $x\geq0$.
\end{cor}
\begin{proof}
  Let $\{e_n\}_{n\in\mathbb Z}$ be the Riesz basis from Proposition \ref{p:Riesz basis on L2}, $V$ the linear vector space generated by $\{e_n\}_{n\in\mathbb Z}$ 
(which is in fact $\ran(\mathcal{A})$) and $\mathcal P_{\mathcal{A}}$ the projector from that proposition. Then $\{(1,0),\{(0,e_n)\}_{n\in\mathbb Z}\}$ is a Riesz basis of $\mathbb C\times V$. Furthermore, $\{g_*,\{g_n\}_{n\in\mathbb Z}\}$ is a Riesz basis of $\Theta^{-1}(\mathbb C\times V)$ because $g_*=\Theta^{-1}(1,0)$ and $g_n=\Theta^{-1}(0,e_n)$. Define $\Pi:=\Theta^{-1}(\Id,\mathcal P_{\mathcal{A}})\Theta$. Then $\Pi$ is a linear projector with the same bound as $\mathcal P_{\mathcal A}$ where
$$
(\Id,\mathcal P_{\mathcal A})(z,f):=(z,\mathcal P_{\mathcal A}f),\quad z\in\mathbb C,f\in L^2(\mathbb R_+,\mathbb C)\,. 
$$
Let $h\in H_\alpha$. Observe that for any $x\in [0,T]$, $\mathrm{cut}(y)=y$ when $0\leq y\leq x$. We have from the definition of the various operators that
\begin{align*}
\Pi h(x)&=\Theta^{-1}(\mathrm{Id},\mathcal P_{\mathcal{A}})(h(0),\exp(\alpha\cdot/2)h') (x)\\
&=\Theta^{-1}\left((h(0),(\exp((\lambda+\alpha/2)\cdot)h')\vert_{[0,T)}(\mathrm{cut}(\cdot)\exp(-\lambda\cdot))\right)(x) \\
&=h(0)+\int_0^x e^{-(\lambda+\alpha/2)y} e^{(\lambda+\alpha/2)\mathrm{cut}(y)}
h'(\mathrm{cut}(y))\,dy \\
&=h(0)+\int_0^xh'(y)\,dy=h(x)\,.
\end{align*}
Hence, $\Pi h(x)=h(x)$ for any $x\in[0,T]$.
\end{proof}
We remark in passing that trivially $g_*^*=g_*$.
In the next proposition we compute the action of the shifting semigroup $\{\U_t\}_{t\geq0}$ on the Riesz basis of Corollary \ref{k:Riesz basis on H} and the dual semigroup on the biorthogonal system.
\begin{prop}\label{p:U and the riesz basis}
For the Riesz basis $\{g_*,\{g_n\}_{n\in\mathbb Z}\}$ in \eqref{e:g-star-def}-\eqref{e:g-n-def} and its biorthogonal system $\{g_*^*,\{g_n^*\}_{n\in\mathbb Z}\}$
derived in Corollary~\ref{k:Riesz basis on H}, it holds
   \begin{enumerate}
    \item $\mathcal U_t g_n = e^{\lambda_n t}g_n + g_n(t)g_*$ and
    \item $\mathcal U^*_tg_n^* =  e^{\overline{\lambda_n} t}g_n^*$,
   \end{enumerate}
   for any $n\in\mathbb Z$.
\end{prop}
\begin{proof}
Claim (1) follows from a straightforward computation.
For claim (2), we compute
   \begin{align*}
     \mathcal U_t^*g_n^* &= g_*\<\mathcal U_t^*g_n^*,g_*\>_{\alpha} + \sum_{k\in\mathbb Z} g^*_k\<\mathcal U_t^*g_n^*,g_k\>_{\alpha} \\
                     &= g_*\< g_n^*,\mathcal U_tg_*\>_{\alpha} + \sum_{k\in\mathbb Z} g^*_k\<g_n^*,\mathcal U_tg_k\>_{\alpha} \\
                     &= e^{\overline{\lambda_n}t}g_n^*
   \end{align*}
for any $n\in\mathbb Z$, $t\geq 0$. Thus, the Proposition follows.
\end{proof}

A certain Lie commutator plays a crucial role in comparing projected solutions to the SPDE~\eqref{e:HJMM-equation} with solutions to 
the approximation. In the next proposition, we derive the essential results for convergence which will be used in the next
Section to analyse approximations of the SPDE~\eqref{e:HJMM-equation}. 
\begin{prop}\label{l:commutator of U and projectors}
 Let $k\in\mathbb N$, $t\geq0$, $H^{T}_\alpha$ be the closed subspace of $H_\alpha$ generated by the Riesz basis $\{g_*,\{g_n\}_{n\in\mathbb Z}\}$ 
defined in \eqref{e:g-star-def}-\eqref{e:g-n-def} with biorthogonal system 
$\{g_*^*,\{g_n^*\}_{n\in\mathbb Z}\}$ given in Corollary~\ref{k:Riesz basis on H}. 
Define the projection
    $$ \Pi_k:H^{T}_\alpha\rightarrow \text{span}\{g_*,g_{-k},\dots,g_k\},h\mapsto h(0)g_* + \sum_{n=-k}^k g_n\<h,g_n^*\>_{\alpha}, $$
    $c_{k,t}:=\sum_{\vert n\vert >k} g_n(t)g_n^*$, and the operator
    $$ \mathcal C_{k,t}:H^{T}_\alpha\rightarrow \text{span}\{g_*\}, h\mapsto \<h,c_{k,t}\>_{\alpha}g_*.$$
   Then, $\|\Pi_k\|_{\text{op}}$ is bounded uniformly in $k$, $\Pi_kh\rightarrow h$, $\sup_{s\in[0,t]} \Vert\mathcal C_{k,s}h\Vert_\alpha\rightarrow 0$ for $k\rightarrow\infty$ and any $h\in H_\alpha^{T}$, and $[\Pi_k,\mathcal U_t] =\mathcal C_{k,t}$. Here, $[\Pi_k,\mathcal U_t]$ denotes the Lie commutator of $\Pi_k$ and $\mathcal U_t$, that is $[\Pi_k,\mathcal U_t]=\Pi_k\mathcal U_t-\mathcal U_t\Pi_k$. 
   
 Moreover, let $X$ be a stochastic process with values in $H_\alpha^{T}$ 
 such that $X(t)=Y(t)+M(t)$ for some square integrable process $Y$ of finite variation and a square integrable martingale $M$. Then,
   $$ \lim_{k\rightarrow\infty}\int_0^t\mathcal C_{k,t-s}dX(s)=0\,,$$
 where the convergence is in $L^2(\Omega,H_\alpha)$.\footnote{$L^2(\Omega,H_\alpha)$ denotes the space of $H_{\alpha}$-valued random variables $Z$ with $\E[\Vert Z\Vert_\alpha^2]<\infty$.}
\end{prop}
\begin{proof}
 Let $h\in H^T_\alpha$. Since $\{g_*,\{g_n\}_{n\in\mathbb Z}\}$ is a Riesz basis of $H^T_\alpha$ we have
$$ 
h = g_*\<h,g_*\>_{\alpha} + \sum_{n\in\mathbb Z} g_n\<h,g_n^*\>_{\alpha}\,,
$$
and hence we get $\Pi_kh\rightarrow h$ for $k\rightarrow \infty$. 

We prove that the operator norm of $\Pi_k$ is uniformly bounded in $k\in\mathbb{N}$. 
Recall from Corollary~\ref{k:Riesz basis on H} and \eqref{def:Theta-inv-op} 
$g_n=\Theta^{-1}(0,\mathcal{A}b_n), n\in\mathbb{Z}$ and $g_*=\Theta^{-1}(1,0)$, where
$\mathcal{A}$ is defined in \eqref{def-A-operator} and $\{b_n\}_{n\in\mathbb{Z}}$ is 
an orthonormal basis of $L^2([0,T],\mathbb{C})$. Without loss of generality, we assume 
$h(0)=0$ for $h\in H_{\alpha}^T$, and find that 
$$
\Pi_kh=\sum_{n=-k}^kg_n\langle h,g_n^*\rangle_{\alpha}=\sum_{n=-k}^k
\mathcal{T}b_n(\mathcal{T}^{-1}h,b_n)_2=\mathcal {T}\sum_ {n=-k}^kb_n(\mathcal{T}^{-1}h,b_n)_2\,.
$$
Here, $\mathcal{T}f:=\Theta^{-1}(0,\mathcal{A}f)\in H_{\alpha}$ for 
$f\in L^2([0,T],\mathbb{C})$, which is a bounded linear operator. Hence,
since $\sum_{n=-k}^kb_n(\mathcal{T}^{-1}h,b_n)_2$ is the projection of 
$\mathcal{T}^{-1}h\in L^2([0,T],\mathbb{C})$ down to its first $2k+1$ coordinates, 
$$
\|\Pi_kh\|_{\alpha}\leq|\mathcal{T}\|_{\text{op}}\left|\sum_{n=-k}^kb_n(\mathcal{T}^{-1}h,b_n)_2\right|_2\leq\|\mathcal{T}\|_{\text{op}}|\mathcal{T}^{-1}h|_2
$$ 
But since $\mathcal{T}^{-1}$ also is a bounded operator, it follows that 
$\|\Pi_k\|_{\text{op}}\leq\|\mathcal{T}\|_{\text{op}}\|\mathcal{T}^{-1}\|_{\text{op}}$.

 Benth and Kr\"uhner~\cite[Lemma 3.2]{benth.kruehner.14} yields that convergence in $H_\alpha$ implies local uniform convergence. Thus, as we know $h-\Pi_kh \rightarrow 0$,
it holds
   $$ \sup_{s\in[0,t]} \vert h(s)-\Pi_{k}h(s)\vert \rightarrow 0\,, $$
  for $k\rightarrow \infty$. Hence, we find 
   $$\sup_{s\in[0,t]}\left\vert\sum_{\vert n\vert>k}g_n(s)\<h,g^*_n\>_{\alpha}\right\vert = \sup_{s\in[0,t]} \vert h(s)-\Pi_{k}h(s)\vert \rightarrow 0\,, $$
   for $k\rightarrow\infty$. Therefore, $\sup_{s\in[0,t]}\Vert \C_{k,s}h\Vert_\alpha\rightarrow 0$ for $k\rightarrow\infty$.
   
   Let $n\in\mathbb Z$. Then, by Proposition~\ref{p:U and the riesz basis}
   \begin{align*}
     [\Pi_k,\U_t]g_n &= \Pi_k(e^{\lambda_n t}g_n + g_n(t)g_*) - 1_{\{\vert n\vert\leq k\}}\U_tg_n \\
                     &= 1_{\{\vert n\vert\leq k\}}e^{\lambda_n t}g_n + g_n(t)g_* - 1_{\{\vert n\vert\leq k\}}(e^{\lambda_n t}g_n + g_n(t)g_*) \\
                     &= 1_{\{\vert n\vert>k\}}g_n(t)g_* \\
                     &= \C_{k,t}g_n
   \end{align*}
   for any $t\geq0$. Moreover,
   \begin{align*}
     [\Pi_k,\U_t]g_* = \Pi_kg_* - \U_tg_* = 0 = \C_{k,t}g_*.
   \end{align*}
   
   Let $\<\<M,M\>\>(t) = \int_0^t Q_sd\<M,M\>(s)$ be the quadratic variation processes of the martingale $M$ given in 
Peszat and 
Zabczyk~\cite[Theorem 8.2]{peszat.zabczyk.07}\footnote{In Peszat and Zabczyk~\cite{peszat.zabczyk.07}, $\<\<\cdot,\cdot\>\>$ is called the operator angle bracket process, while $\<\cdot,\cdot\>$ is the angle bracket process.}. Then, Peszat and Zabczyk~\cite[Theorem 8.7(ii)]{peszat.zabczyk.07} yields
   \begin{align*}
 \E\left(\Vert \int_0^t \C_{k,t-s}dM(s) \Vert_\alpha^2\right) &= \E \int_0^t \tr(\C_{k,t-s}Q_s\C^*_{k,t-s})d\<M,M\>(s)\,.    
   \end{align*}
Recall that for $h\in H_{\alpha}^T$, we find $\C_{k,t}h=\<h,c_{k,t}\>_{\alpha}g_*$. Thus,
$$
\<h,\C_{k,t}^*g_*\>_{\alpha}=\<\C_{k,t}h,g_*\>_{\alpha}=\<h,c_{k,t}\>_{\alpha}\,,
$$
which gives that $\C_{k,t}^*g_*=c_{k,t}$. For $g\in H_\alpha^T$ orthogonal to $g_*$ we have
  $$ \<h,\C_{k,t}^*g\>_\alpha =\<\C_{k,t}h,g\>_{\alpha}=\<h,c_{k,t}\>_{\alpha}\<g_*,g\>_{\alpha}= 0 $$
 for any $h\in H_\alpha^T$ and hence $\C_{k,t}^*g=0$. We get
\begin{align*}
\tr(\C_{k,t-s}Q_s\C^*_{k,t-s}) &= \<\C_{k,t-s}Q_s\C_{k,t-s}^*g_*,g_*\>_{\alpha} \\
&=\<Q_sc_{k,t-s},c_{k,t-s}\>_{\alpha} \\
&\leq \Vert c_{k,t-s}\Vert_\alpha^2 \tr(Q_s)\,. 
\end{align*}
Hence,
  \begin{align*}
 \E\left(\left\Vert \int_0^t \C_{k,t-s}dM(s) \right\Vert_\alpha^2\right) &= \E \int_0^t \tr(\C_{k,t-s}Q_s\C^*_{k,t-s})d\<M,M\>(s) \\
               &\leq \sup_{s\in[0,t]}\Vert c_{k,s}\Vert_\alpha^2  \E\left( \int_0^t\tr(Q_s)d\<M,M\>(s) \right) \\
               &= \sup_{s\in[0,t]}\Vert c_{k,s}\Vert_\alpha^2 \E\left(\Vert M(t)-M(0)\Vert_\alpha^2\right) \\
               &\rightarrow 0
   \end{align*}
   for $k\rightarrow\infty$. Similarily, we get
   $$ \left\Vert \int_0^t \C_{k,t-s} dY(s) \right\Vert_\alpha^2 \leq \sup_{s\in[0,t]} \Vert c_{k,s}\Vert_\alpha^2\left(\int_0^t \Vert dY\Vert_\alpha(s)\right)^2 \rightarrow 0$$
   as $k\rightarrow 0$, where $\Vert dY\Vert_\alpha$ denotes the total variation measure associated with $dY$ (see Dinculeanu~\cite[Definition \S 2.1]{dinculeanu.00}). The claim follows.
\end{proof}
The projection operator $\Pi_k$ plays an important role in the arbitrage-free approximation of the forward term structure. For notational convenience,
we denote 
\begin{equation}
\label{def-H-k-space}
H_{\alpha}^{T,k}:=\text{span}\{g_*,g_{-k},\dots,g_k\}\,,
\end{equation}
for any $k\in\mathbb N$. From the above considerations, we have that $\Pi_k$ projects the space $H_{\alpha}^T$ down to $H_{\alpha}^{T,k}$. 

Our next aim is to identify the convergence speed of approximations in $H_{\alpha}^{T,k}$ of certain smooth
elements $f\in H_{\alpha}^T$, that is, how close is $\Pi_kf$ to $f$ in terms
of number of Riesz basis functions. We show a couple of technical results first.
\begin{cor}\label{k:distance to ONB}
  Let $f\in H_\alpha^T$. Then, we have
   \small{$$ \frac{e^{-2\lambda T}}{1-e^{-2\lambda T}} \left(|f(0)|^2 + \sum_{n\in\mathbb Z} \vert\<f,g_n^*\>_{\alpha}\vert^2\right)\leq \Vert f\Vert_\alpha^2 \leq \frac{1}{1-e^{-2\lambda T}} \left(|f(0)|^2 + \sum_{n\in\mathbb Z} \vert\<f,g_n^*\>_{\alpha}\vert^2\right)\,.$$}
\end{cor}
\begin{proof}
  Corollary~\ref{k:Riesz basis on H} states that $\{g_*,\{g_n\}_{n\in\mathbb Z}\}$ is a Riesz basis of $H_\alpha^T$. Moreover, it is given by $g_*=\Theta^{-1}(1,0)$, $g_n=\Theta^{-1}(0,e_n)$ for any $n\in\mathbb Z$ where $\Theta$ is the isometry given in \eqref{def:Theta-inv-op} and $\{e_n\}_{n\in\mathbb Z}$ is the Riesz basis given in Proposition \ref{p:Riesz basis on L2}. Moreover, Lemma~\ref{l:stetige Einbettung} yields that $e_n=\mathcal Ab_n$ for any $n\in\mathbb Z$ where $\{b_n\}_{n\in\mathbb Z}$ is an orthonormal basis of 
$L^2([0,T],\mathbb C)$ and $\Vert \mathcal A\Vert_{\mathrm{op}}^2\leq \frac{1}{1-e^{-2\lambda T}}$. Thus, we can construct a Hilbert space with orthonormal basis $\{b_*,\{b_n\}_{n\in\mathbb Z}\}$ and a bounded linear operator
$\mathcal B$ with $\Vert \mathcal B\Vert_{\mathrm{op}}^2\leq \frac{1}{1-e^{-2\lambda T}}$, such that $g_*=\mathcal Bb_*$, $g_n=\mathcal Bb_n$. 
Thus, we have
   \begin{align*}
      \Vert f\Vert_\alpha^2  &= \Vert g_*\<f,g_*\>_{\alpha} + \sum_{n\in\mathbb Z}g_n\<f,g_n^*\>_{\alpha} \Vert_\alpha^2 \\
                             &= \Vert \mathcal Bb_*\<f,g_*\>_{\alpha} + \sum_{n\in\mathbb Z}\mathcal Bb_n\<f,g_n^*\>_{\alpha} \Vert_\alpha^2 \\
                             &\leq \frac{1}{1-e^{-2\lambda T}} \left(|\<f,g_*\>_{\alpha}|^2 + \sum_{n\in\mathbb Z}|\<f,g_n^*\>_{\alpha}|^2 \right) \\
   \end{align*}
   where $\{g_*,\{g^*_n\}_{n\in\mathbb Z}\}$ denotes the biorthogonal system to $\{g_*,\{g_n\}_{n\in\mathbb Z}\}$ given in Corollary~\ref{k:Riesz basis on H}.
The lower inequality simply uses the lower inequality of Lemma~\ref{l:stetige Einbettung} instead.
\end{proof}
The next technical result connects the inner product of elements in $H_{\alpha}^T$ with the 
biorthogonal basis functions to a simple
Fourier-like integral on $[0,T]$:
\begin{cor}
\label{cor:alpha-2-inner-prod}
Assume $f\in H_{\alpha}^T$. Then, for any $n\in\mathbb Z$, 
$$
\<f,g_n^*\>_{\alpha}=(1-e^{-2\lambda T})^{-1}T^{-1/2}\int_0^Tf'(x)\exp\left((-\frac{2\pi i}{T}n-\lambda+\frac{\alpha}{2})x\right)\,dx
$$
\end{cor}
\begin{proof}
First, recall that $g_n^*=\Theta^*(0,e_n)$ for $n\in\mathbb Z$, where $\Theta$ is defined 
in the \eqref{def:Theta-inv-op}. Thus,
\begin{align*}
\<f,g_n^*\>&=\<f,\Theta^*(0,e_n)\>_{\alpha} \\
&=(\Theta f,(0,e_n))_{\mathbb C\times L^2(\mathbb R_+)} \\
&=((f(0),e^{\alpha\cdot/2}f'),(0,e_n))_{\mathbb C\times L^2(\mathbb R_+)} \\
&=(e^{\alpha\cdot/2}f',e_n)_2\,.
\end{align*}
Note that $\exp(\alpha\cdot/2)f'$ and $e_n=\mathcal A b_n$ are elements of $\mathrm{ran}(\mathcal A)$. 
If $h\in\mathrm{ran}(\mathcal{A})$, then there exists a $\hat{h}\in L^2([0,T],\mathbb C)$ such that
$h=\mathcal A \hat{h}$, or, $h(x)=\exp(-\lambda x)\hat{h}(\mathrm{cut}(x))$. Observe that 
for $x\in[0,T]$, $\hat{h}(x)=\exp(\lambda x)h(x)$. Then, if $g\in\mathrm{ran}(\mathcal{A})$, we find
\begin{align*}
(h,g)_2&=\int_0^{\infty}h(x)\overline{g(x)}\,dx \\
&=\int_0^{\infty}e^{-2\lambda x}\hat{h}(\mathrm{cut}(x))\overline{\hat{g}(\mathrm{cut}(x)}\,dx \\
&=\sum_{n=0}^{\infty}e^{-2\lambda n T}\int_{nT}^{(n+1)T}e^{-2\lambda(x-nT)}\hat{h}(\mathrm{cut}(x))
\overline{\hat{g}(\mathrm{cut}(x))}\,dx \\
&=\sum_{n=0}^{\infty}e^{-2\lambda n T}\int_0^Te^{-2\lambda x}\hat{h}(x)\overline{\hat{g}(x)}\,dx \\
&=(1-e^{-2\lambda T})^{-1}\int_0^Th(x)\overline{g(x)}\,dx\,.
\end{align*}
Thus,
\begin{align*}
\<f,g_n^*\>&=(1-e^{-2\lambda T})^{-1}\int_0^Te^{\alpha x/2}f'(x)\overline{e_n(x)}\,dx \\
&=(1-e^{-2\lambda T})^{-1}T^{-1/2}\int_0^Tf'(x)\exp\left((-\frac{2\pi i}{T}n-\lambda+\frac{\alpha}{2})x\right)\,dx
\end{align*}
Hence, the result follows.
\end{proof}

With this results at hand, we can prove a convergence rate of order $1/k$ for sufficiently smooth functions in $H_{\alpha}^T$.
\begin{prop}\label{p:approximation speed}
Assume $f\in H_\alpha^{T}$ is such that $f\vert_{[0,T]}$ is twice continuously differentiable. 
 Then, we have
$$ 
\left\Vert f - \Pi_kf\right\Vert^2_\alpha \leq\frac{C_1}{k}\,,
$$
for any $k\in\mathbb N$, where
$$
C_1=\frac{T\left\vert f'(T)e^{T(-\lambda+\alpha/2)}-f'(0)\right\vert^2+(\int_0^T \vert f''(x)\vert e^{x(-\lambda+\alpha/2)}\,dx)^2}{\pi^2(1-e^{-2\lambda T})^3}\,,
$$  
and we recall the projection operator $\Pi_k$ from Proposition~\ref{l:commutator of U and projectors}.
\end{prop}
\begin{proof}
 Corollary~\ref{k:distance to ONB} yields
  \begin{align*}
  \Vert f - \Pi_kf\Vert^2_\alpha &= \Vert \sum_{\vert n\vert>k} g_n\<f,g_n^*\>_{\alpha}\Vert_\alpha^2
       \leq C \sum_{\vert n\vert>k} \vert \<f,g_n^*\>_{\alpha}\vert^2\,
  \end{align*}
  where $C:=(1-e^{-2\lambda T})^{-1}$. Define $h_n(x):=\exp(\xi_nx)$, $x\geq 0$, where we 
denote $\xi_n=-\frac{2\pi i}{T}n-\lambda+\frac{\alpha}{2}$. Then, by Corollary~\ref{cor:alpha-2-inner-prod}
and integration-by-parts we find
 \begin{align*}
  \vert \<f,g_n^*\>_{\alpha}\vert^2 &=C^2T^{-1}\left\vert\int_0^Tf'(x)h_n(x)dx\right\vert^2 \\
     &= C^2T^{-1}\frac{1}{\vert\xi_n\vert^2}\left\vert f'(T)h_n(T)-f'(0)h_n(0)-\int_0^T f''(x)h_n(x)\,dx \right\vert^2 \\
     & \leq \frac{2C^2}{T}\frac{1}{\vert\xi_n\vert^2}A_f\,,
 \end{align*}
  for any $n\in\mathbb Z\backslash\{0\}$, where the constant $A_f$ is
$$
A_f:=\left\vert f'(T)e^{T(-\lambda+\alpha/2)}-f'(0)\right\vert^2+(\int_0^T \vert f''(x)e^{x(\lambda-\alpha/2)}\, dx)^2\,.
$$
Moreover, we have
 \begin{align*}
    \sum_{\vert n\vert>k}\frac{1}{\vert\xi_n\vert^2}= 2 \sum_{n>k}\frac{1}{\vert\xi_n\vert^2}\leq \frac{T^2}{2\pi^2 k}.
 \end{align*}
  Putting the estimates together, we get
$$
\Vert f - \Pi_kf\Vert^2_\alpha \leq A_f\frac{C^3T}{\pi^2k}\,, 
$$
as claimed.
\end{proof}
We can find a similar convergence rate for $c_{k,t}$, a result which becomes useful later:
\begin{lem}
\label{lemma:approximation speed_ckt}
Let $c_{k,t}$ be given as in Proposition~\ref{l:commutator of U and projectors}. Then,
$$
\Vert c_{k,t}\Vert_\alpha^2 \leq \frac{C_2}{k}\,,
$$
for any $k\in\mathbb N$, where $C_2=T/\pi^2(1-\exp(-2\lambda T))$.
\end{lem}
\begin{proof}
We appeal to Corollary~\ref{k:distance to ONB}, using $\{g_n^*\}_{n\in\mathbb Z}$ as the Riesz basis
with biorthogonal system $\{g_n\}_{n\in\mathbb Z}$, to find
\begin{align*}
\|c_{k,t}\|_{\alpha}^2&=\|\sum_{|n|>k}g_n(t)g_n^*\|_{\alpha}^2 \\
&\leq C\sum_{|n|>k}|g_n(t)|^2 \\
&=\frac{C}{T}\sum_{|n|>k}\frac{1}{\vert\lambda_n\vert^2}\left\vert e^{\lambda_n t}-1\right\vert^2 \\
&\leq\frac{2C}{T}(1+e^{-(2\lambda+\alpha)t})\sum_{|n|>k}\frac{1}{\vert\lambda_n\vert^2} \\
&\leq \frac{CT}{\pi^2}\frac1k\,,
\end{align*}
for $C=(1-\exp(-2\lambda T))^{-1}$. Hence, the result follows.
\end{proof}

With these results we are now in the position to study arbitrage-free approximations of the forward dynamics in \eqref{e:HJMM-equation}.

\section{Arbitrage free approximation of forward term structure models}\label{s:approximation}
In this section we find an arbitrage-free approximation of a forward term structure model -- stated in the Heath-Jarrow-Morton-type setup -- which lives in a 
finite dimensional state space. We furthermore derive the convergence speed of the approximation, and extend the results to account for 
forward contracts delivering the underlying commodity over a period which is the case for electricity and gas.

Consider the SPDE \eqref{e:HJMM-equation} with a mild solution $f\in H_{\alpha}$ given by \eqref{e:HJMM-equation-mild}. We recall 
from \eqref{e:g-star-def}-\eqref{e:g-n-def} and Corollary~\ref{k:Riesz basis on H} the 
Riesz basis $\{g_*,\{g_n\}_{n\in\mathbb Z}\}$ on the space $H_{\alpha}^T$ with the biorthogonal system  
$\{g_*,\{g_n^*\}_{n\in\mathbb Z}\}$. Furthermore,
$\Pi$ is the projection of $H_{\alpha}$ on $H_{\alpha}^T$, while from \eqref{def-H-k-space} and 
Proposition~\ref{p:U and the riesz basis} we have the projection $\Pi_k$ of $H_{\alpha}^T$
on $H_{\alpha}^{T,k}$ and the operator $\mathcal C_{k,t}$ for $k\in\mathbb N$, $t\geq 0$. 
Let us define the continuous linear operator  $\Lambda_k:H_{\alpha}\rightarrow H_{\alpha}^{T,k}$ by
\begin{equation}
\Lambda_k= \Pi_k\Pi
\end{equation}
for any $k\in\mathbb N$. The following theorem is one of the main results of the paper:

\begin{thm}\label{t:main statement}
For $k\in\mathbb N$, let $f_k$ be the mild solution of the SPDE
\begin{equation}
\label{e:approx-f-k}
df_k(t) = \partial_x f_k(t) dt + \Lambda_k\beta(t) dt + \Lambda_k\Psi(t) dL(t),\quad t\geq 0, f_k(0) = \Lambda_kf_0\,.
\end{equation}
Then, we have
 \begin{enumerate}
  \item $\E\left[\sup_{x\in[0,T-t]}\vert f_k(t,x) - f(t,x)\vert^2\right] \rightarrow 0$ for $k\rightarrow\infty$ and any $t\in [0,T]$,
  \item $f_k$ takes values in the finite dimensional space $H_\alpha^{T,k}$, moreover, $f_k$ is a strong solution to the SPDE \eqref{e:approx-f-k}, i.e.\ $f_k\in\dom(\partial_x)$, {$t\mapsto \partial_xf_k(t)$} is $P$-a.s.\ Bochner-integrable and
   $$ f_k(t) = f_k(0) + \int_0^t (\partial_xf_k(s)+\Lambda_k\beta(s))ds + \int_0^t \Lambda_k\Psi(s) dL(s)\,, $$
   \item and,
    \begin{align*}
      f_k(t) &= S_k(t) + \sum_{n=-k}^k \left(e^{\lambda_n t} \<f_k(0),g_n^*\>_{\alpha} + \int_0^t e^{\lambda_n (t-s)}dX_n(s)\right)g_n \,,
    \end{align*} 
    where $S_k(t)= \delta_0(f_k(t))$ and $X_n(t) := \int_0^t \<\Pi\beta(s)ds+\Pi\Psi(s)dL(s),g_n^*\>_{\alpha}$ for any $n\in\mathbb Z$, $t\geq0$.
 \end{enumerate}
\end{thm}
\begin{proof}
 (1) Define 
 $$f_\Pi(t) := \U_t\Pi f_0 + \int_0^t \U_{t-s}(\Pi\beta(s)ds + \Pi\Psi(s)dL(s))),\quad t\geq0.$$
Since $f_k$ is a mild solution,  we have
 \begin{align*}
   f_k(t) &= \U_t\Pi_k\Pi f_0 + \int_0^t \U_{t-s}\Pi_k(\Pi\beta(s)ds + \Pi\Psi(s)dL(s)) \\
          &= \Pi_k\U_t\Pi f_0 + \int_0^t \Pi_k\U_{t-s}(\Pi\beta(s)ds + \Pi\Psi(s)dL(s)) \\
          &\hspace{0.3cm} - \mathcal C_{k,t}\Pi f_0 - \int_0^t \mathcal C_{k,t-s}(\Pi\beta(s)ds + \Pi\Psi(s)dL(s)) \\
          &= \Pi_k\left(\U_t\Pi f_0 + \int_0^t \U_{t-s}(\Pi\beta(s)ds + \Pi\Psi(s)dL(s)))\right) \\
          &\hspace{0.3cm} - \mathcal C_{k,t}\Pi f_0 - \int_0^t \mathcal C_{k,t-s}(\Pi\beta(s)ds + \Pi\Psi(s)dL(s)) \\
          &= \Pi_k(f_\Pi(t)) - \mathcal C_{k,t}\Pi f_0 - \int_0^t \mathcal C_{k,t-s}(\Pi\beta(s)ds + \Pi\Psi(s)dL(s))
 \end{align*}
 for any $t\geq0$. From Benth and Kr\"uhner~\cite[Lemma 3.2]{benth.kruehner.14} the sup-norm is dominated by the $H_{\alpha}$-norm. Thus, there is a constant $c>0$ such that
  \begin{align*}
   \E\left[\sup_{x\in[0,T-t]}\vert \Pi_k(f_\Pi(t,x)) - f_\Pi(t,x)\vert^2\right] & \leq c\E\left[ \Vert (\Pi_k-\mathcal I)f_\Pi(t) \Vert^2_\alpha \right]
  \end{align*}
 for any $t\geq 0$ where $\mathcal I$ denotes the identity operator on $H_\alpha$. The dominated convergence theorem yields that the right-hand side converges to $0$ for $k\rightarrow \infty$. Clearly, we have
   $$ \sup_{x\in[0,T-t]}\vert \mathcal C_{k,t}f_\Pi(0,x) \vert \leq c\Vert \mathcal C_{k,t}f_\Pi(0) \Vert_\alpha \rightarrow 0\,,  $$
 for $k\rightarrow \infty$. Proposition~\ref{l:commutator of U and projectors} states that
   $$ \E\left\Vert \int_0^t \mathcal C_{k,t-s}(\Pi\beta(s)ds + \Pi\Psi(s)dL(s)) \right\Vert_\alpha^2 \rightarrow 0\,,$$
  for $k\rightarrow 0$. Hence, we have
 $$ \E\left(\sup_{x\in[0,T-t]}\vert f_k(t,x)-f_\Pi(t,x) \vert^2\right) \rightarrow 0\,, $$
 for $k\rightarrow \infty$ and any $t\in [0,T]$. Since $f_\Pi(t,x) = f(t,x)$ for any $t\in[0,T]$, $x\in [0,T-t]$ the first part follows.
 
 (2) Note first that $\partial_x g_n(x)=\exp(\lambda_n x)/\sqrt{T}=\lambda_ng_n(x)+g_*(x)/\sqrt{T}$, and hence
$\partial_x g_n\in H_{\alpha}^{T,k}$ whenever $|n|\leq k$. Thus, $H_{\alpha}^{T,k}$ is invariant under the
generator $\partial_x$, and its restriction to $H_{\alpha}^{T,k}$ is continuous and bounded. We find that $f_k$ takes values only in $H_\alpha^{T,k}$ because 
 \begin{align*}
f_k(t) &= \Pi_k\left(\U_t\Pi f_0 + \int_0^t \U_{t-s}(\Pi\beta(s)ds + \Pi\Psi(s)dL(s)))\right) \\
          &\hspace{0.3cm} - \mathcal C_{k,t}\Pi f_0 - \int_0^t \mathcal C_{k,t-s}(\Pi\beta(s)ds+\Pi\Psi(s)dL(s))\,,
 \end{align*} 
where all summands are clearly in $H_\alpha^{T,k}$. 

 (3) As $f_k(t)\in H_{\alpha}^{T,k}$, we have the representation
$$
f_k(t)=\<f_k(t),g_*^*\>_{\alpha}g_*+\sum_{n=-k}^k \<f_k(t),g_n^*\>_{\alpha}g_n\,.
$$ 
Since $g_*^*=1$, we find that $\<f_k(t),g_*^*\>_{\alpha}=f_k(t,0)=\delta_0(f_k(t))$. Thus, from the mild solution of 
\eqref{e:approx-f-k} we find, using Proposition \ref{p:U and the riesz basis}
\begin{align*}
f_k(t)&=S_k(t)+\sum_{n=-k}^k \left\<\mathcal{U}_t f_k(0)+\int_0^t\U_{t-s}(\Lambda_k\beta(s)ds + \Lambda_k\Psi(s)dL(s)),g_n^*\right\>_{\alpha}g_n \\
&=S_k(t)+\sum_{n=-k}^{k}\<f_k(0),\mathcal{U}_t^*g_n^*\>_{\alpha}g_n \\
&\qquad\qquad+\sum_{n=-k}^{k}\int_0^t\<\Lambda_k\beta(s)ds + \Lambda_k\Psi(s)dL(s),\mathcal U_{t-s}^*g_n^*\>_{\alpha}g_n \\
&=S_k(t)+\sum_{n=-k}^{k}e^{\lambda_n t}\<f_k(0),g_n^*\>_{\alpha}g_n \\
&\qquad\qquad+\sum_{n=-k}^k\int_0^t e^{\lambda_n(t-s)}\<\Lambda_k\beta(s)ds + \Lambda_k\Psi(s)dL(s),g_n^*\>_{\alpha}g_n \,.
\end{align*}
Observe that for any $f\in H_{\alpha}$,
$$
\Lambda_k f=\Pi_k(\Pi f)=(\Pi f)(0)g_*+\sum_{m=-k}^k\<\Pi f,g_m^*\>_{\alpha}g_m\,,
$$
and since $\{g_*,\{g_n\}_{n\in\mathbb Z}\}$, $\{g_*^*,\{g^*_n\}_{n\in\mathbb Z}\}$ are biorthogonal systems 
\begin{align*}
\<\Lambda_k f,g_n^*\>_{\alpha}&=(\Pi f)(0)\<g_*,g_n^*\>_{\alpha}+\sum_{m=-k}^k\<\Pi f,g_m^*\>_{\alpha}\<g_m,g_n^*\>_{\alpha}
=\<\Pi f,g_n^*\>_{\alpha}1_{\{\vert n\vert\leq k\}}\,.
\end{align*}
Hence, the claim follows.
\end{proof}

Another view on Theorem \ref{t:main statement} is that all processes in the $k$-th approximation 
of $f$ can be expressed in terms of the factor processes $X_*,X_{-k},\dots,X_k$, as stated below.
\begin{cor}\label{k:state variables}
 Under the assumptions and notations of Theorem \ref{t:main statement}, we have for $k\in\mathbb N$,
  \begin{align*}
    f_k(t,x) &= S_k(t) + \sum_{n=-k}^k U_n(t) g_n(x)\,,
  \end{align*}
 for any $0\leq t<\infty$ and $x\geq 0$. Here,
\begin{align*}
     S_k(t)&= S_k(0) + X_*(t) + \sum_{n=-k}^k \left( g_n(t)U_n(0) +\int_0^tg_n(t-s)dX_n(s)\right) \,,
  \end{align*}
with,
  \begin{align*}
     X_n(t) &:= \left\< \int_0^t(\Pi\beta(s)ds+\Pi\Psi(s)dL(s)),g_n^*\right\>_{\alpha}, \\
     X_*(t) &:= \left\< \int_0^t(\Pi\beta(s)ds+\Pi\Psi(s)dL(s)),g_*\right\>_{\alpha}, \\
    U_n(t) &:= e^{\lambda_n t} \<f_k(0),g_n^*\> + \int_0^t e^{\lambda_n (t-s)}dX_n(s)   \end{align*}
 for $n\in\{-k,\dots, k\}$.
\end{cor}
\begin{proof}
  The first equation is a restatement of (3) in Theorem \ref{t:main statement}. Proposition \ref{p:U and the riesz basis} yields
   $$ \<\U_t h,g_*\>_{\alpha} = \<h,g_*\>_{\alpha} + \sum_{n=-k}^k g_n(t)\<h,g_n^*\>_{\alpha} $$
  for any $h\in H_\alpha^{T,k}$ with $h=\<h,g_*\>_{\alpha}g_*+\sum_{n=-k}^k \<h,g_n^*\>_{\alpha}g_n$. Thus, since $g_*=1$ and $g_n(0)=0$ we have
   \begin{align*}
     S_k(t) &= f_k(t,0) \\
&=\< f_k(t),g_*\>_{\alpha} \\
            &= \< \U_{t}f_k(0),g_*\>_{\alpha} + \int_0^t \<\U_{t-s}(\Lambda_k\beta(s)\,ds+\Lambda_k\Psi(s)\,dL(s)),g_*\>_{\alpha} \\
		&= \<f_k(0),g_*\>_{\alpha}+\sum_{n=-k}^kg_n(t)\<f_k(0),g_n^*\>_{\alpha} \\
		&\qquad\qquad+\int_0^t\<\Lambda_k\beta(s)\,ds+\Lambda_k\Psi(s)\,dL(s),g_*\>_{\alpha} \\ 
		&\qquad\qquad+\sum_{n=-k}^k\int_0^tg_n(t-s)\<\Lambda_k\beta(s)+\Lambda_k\Psi(s)\,dL(s),g_n^*\>_{\alpha}\,.
   \end{align*}
As in the proof of Theorem~\ref{t:main statement}, we have $\<\Lambda_k f,g_n^*\>_{\alpha}=\<\Pi f,g_n^*\>_{\alpha}$ for any $f\in H_{\alpha}$. 
Similarly, $\<\Lambda_k f,g_*\>_{\alpha}=\<\Pi f,g_*\>_{\alpha}$ for $n\in\mathbb Z$ with $\vert n\vert\leq k$. The result follows.
\end{proof}
The processes $S_k, U_{-k},\dots, U_k$ in Corollary~\ref{k:state variables} capture at any time $t$ the whole state of the market in the approximation model. 
I.e., the spot price and the forward curve are simple functions of these state variables. As we will see in Corollary \ref{k:Forward prices} below, the forward prices 
of contracts with delivery periods can be expressed in these state variables as well. Note that if we assume  $\<\Pi\beta,g_n^*\>$, $\<\Pi\Psi,g_n^*\>$ are deterministic and constant, then $(X_{-k},\dots,X_k)$ is a $2k+1$-dimensional L\'evy process and $U_{-k},\dots,U_k$ are Ornstein-Uhlenbeck processes. This corresponds to the spot price model suggested in Benth, Kallsen and Meyer-Brandis~\cite{benth.et.al.05}.

From the proof of Corollary~\ref{k:state variables} we find that $S_k(0)=\<f_k(0),g_*\>_{\alpha}$. But then 
$$
S_k(0)=\<\Lambda_k f_0,g_*\>_{\alpha}=\<\Pi f_0,g_*\>_{\alpha}=(\Pi f_0)(0)=f_0(0)\,.
$$
Obviously, $f_0(0)$ is equal to today's spot price, so we obtain that the starting point of the process $S_k(t)$ in the approximation is today's spot price. 
Furthermore, since we have $f_k(t,0)=S_k(t)$ because $g_n(0)=0$ for all $n\in\mathbb Z$, $S_k(t)$ is the approximative spot price dynamics associated
with $f_k(t)$. For $U_n(0)$, $n\in\mathbb Z$ invoking Corollary~\ref{cor:alpha-2-inner-prod} shows that
\begin{align*}
U_n(0)&=\<\Pi f_0,g_n^*\>_{\alpha} \\
&=\frac1{\sqrt{T}(1-e^{-2\lambda T})}\int_0^{T}(\Pi f_0)'(y)
\exp((-\lambda+\alpha/2)x)\exp\left(\frac{2\pi i}{T}nx\right)\,dy\,.
\end{align*}
This is the Fourier transform of the initial forward curve $f_0$ (or, rather its derivative
scaled by an exponential function). In any case, both $S_k(0)$ and $U_n(0)$ are 
given by (functionals of) the initial forward curve $f_0$.

Next, we would like to identify the convergence speed of our approximation, that is, the rate for the convergence in part (1) of Theorem \ref{t:main statement}.
\begin{prop}\label{p:convergence speed}
Assume that $x\mapsto f(t,x)$ is twice continuously differentiable and let
 $f_k$ be the mild solution of the SPDE
  $$ df_k(t) = \partial_x f_k(t) dt + \Lambda_k\beta(t) dt + \Lambda_k\Psi(t) dL(t),\quad t\geq 0, f_k(0) = \Lambda_kf_0\, .$$
  Then, we have
 $$
\E\left[\sup_{x\in[0,T-t]}\vert f_k(t,x) - f(t,x)\vert^2\right] \leq \frac{A(t)}{k}\,,
$$
  for any $k>1$, where 
\begin{align*}
A(t)&:=\frac{3T(1+\alpha^{-1})}{(1-e^{-2\lambda T})}\left\{\|\Pi f_0\|_{\alpha}^2+\int_0^T \E[\tr(\Psi(s)Q\Psi^*(s))]ds+\left(\int_0^T \E\left[\Vert \beta(s)\Vert_\alpha\right]\,ds\right)^2 \right\} \\
&\qquad+\frac{3(1+\alpha^{-1})}{\pi^2(1-e^{-2\lambda T})^3}\left\{T\E\left[|\partial_xf_{\Pi}(t,T)
e^{T(-\lambda+\alpha/2)}-\partial_xf_{\Pi}(t,0)|^2\right] \right. \\
&\qquad\qquad\left.+\left(\int_0^T\E\left[|\partial^2_xf_{\Pi}(t,x)|\right]e^{x(-\lambda+\alpha/2)}\,dx\right)^2\right\}\,.
\end{align*}
\end{prop}
\begin{proof}
  In the proof of Theorem~\ref{t:main statement} we have shown that
  \begin{align*}
    f_k(t) &= \Pi_k(f_\Pi(t)) - \mathcal C_{k,t}\Pi f_0 - \int_0^t \mathcal C_{k,t-s}(\Pi\beta(s)ds + \Pi\Psi(s)dL(s))\,,
  \end{align*}
  where $f_\Pi(t) := \U_t\Pi f_0 + \int_0^t \U_{t-s}(\Pi\beta(s)ds + \Pi\Psi(s)dL(s)))$ for any $t\geq 0$. By Proposition~\ref{p:approximation speed} we have
  $$ 
\left\Vert f_\Pi(t) - \Pi_k(f_\Pi(t))\right\Vert^2_\alpha \leq\frac{C_1(t)}{k}
$$
where $C_1(t)$ is a random variable defined by 
$$
C_1(t)=\frac{T|\partial_xf_{\Pi}(t,T)e^{T(-\lambda+\alpha/2)}-\partial_xf_{\Pi}(t,0)|^2+(\int_0^T|\partial^2_xf_{\Pi}(t,x)|e^{x(-\lambda+\alpha/2)}\,dx)^2}{\pi^2(1-e^{-2\lambda T})^3}\,.
$$
Remark that from the proof of Theorem~\ref{t:main statement} we find for any $h\in H_{\alpha}^T$
\begin{align*}
\|\mathcal C_{k,t}h\|_{\alpha}^2&=\|\<h,c_{k,t}\>_{\alpha}g_*\|_{\alpha}^2=|\<h,c_{k,t}\>_{\alpha}|^2\leq \|h\|_{\alpha}^2\|c_{k,t}\|_{\alpha}^2\,,
\end{align*}
and therefore, from Lemma~\ref{lemma:approximation speed_ckt}
$$
\|\mathcal C_{k,t}h\|_{\alpha}^2\leq\|h\|_{\alpha}^2\frac{C_2}{k}\,,
$$
for the constant $C_2=T/\pi^2(1-e^{-2\lambda T})$. 
 Then, we have
  \begin{align*}
   \Vert f_k(t) - f_\Pi(t) \Vert_\alpha^2  &\leq 3\Vert \Pi_k(f_\Pi(t))-f_\Pi(t)\Vert_\alpha^2 + 3\Vert \mathcal C_{k,t}\Pi f_0\Vert_\alpha^2 \\
&\qquad\qquad+ 
3\Vert \int_0^t \mathcal C_{k,t-s}(\Pi\beta(s)ds + \Pi\Psi(s)dL(s)) \Vert_\alpha^2 \\
     &\leq \frac{3C_1(t)}{k} + \frac{3C_2}{k}\Vert \Pi f_0\Vert_\alpha^2  \\
&\qquad\qquad+ 3\Vert \int_0^t \mathcal C_{k,t-s}(\Pi\beta(s)ds + \Pi\Psi(s)dL(s)) \Vert_\alpha^2.
  \end{align*}
By Lemma~3.2 in Benth and Kr\"uhner~\cite{benth.kruehner.14}, the supremum norm is bounded by the $H_{\alpha}$-norm with a constant 
$c=\sqrt{1+\alpha^{-1}}$. Hence, taking expectations, yield
    \begin{align*}
    \E&\left[\sup_{x\in[0,T-t]}\vert f_k(t,x) - f(t,x) \vert^2 \right] \\
&\qquad \leq c^2\E\left[ \Vert f_k(t) - f_\Pi(t) \Vert_\alpha^2 \right] \\
      &\qquad\leq \frac{3c^2}{k}\left(\E\left[C_1(t)\right]+ C_2\Vert \Pi f_0\Vert_\alpha^2\right) \\
&\qquad\qquad + \frac{3c^2}{k}C_2\left(\int_0^T \E[\tr(\Psi(s)Q\Psi^*(s))]ds + \left(\int_0^T \E\left[\Vert \beta(s)\Vert_\alpha\right] ds\right)^2 \right) \,.   
    \end{align*}
The result follows.
\end{proof}

In electricity and gas markets forward contracts deliver over a future period rather than at a fixed time. The holder of the forward contract receives a uniform stream of electricity 
or gas over an agreed time period $[T_1,T_2]$. The forward prices of delivery period contracts can be derived from a "fixed-delivery time" forward curve model (see Benth et al.~\cite{BSBK-book})
by
\begin{equation}
\label{e:el-forward-def}
F(t,T_1,T_2) := \frac{1}{T_2-T_1}\int_{T_1}^{T_2} f(t,s-t)\,, ds 
\end{equation}
where $f$ is given by the SPDE \eqref{e:HJMM-equation}. The following Corollary adapts Theorem~\ref{t:main statement} to the case of forward
contracts with delivery period.
\begin{cor}\label{k:Forward prices}
 Assume the conditions of Theorem~\ref{t:main statement} and define
  \begin{align*}
    F_k(t,T_1,T_2) &:= \frac{1}{T_2-T_1}\int_{T_1}^{T_2} f_k(t,s-t) ds
  \end{align*}
 for any $0\leq t \leq T_1\leq T_2\leq T$. Then, we have
  $$ F_k(t,T_1,T_2) \rightarrow F(t,T_1,T_2)$$
  for $k\rightarrow\infty$ in $L^2(\Omega)$ where $F$ is given in \eqref{e:el-forward-def}. Furthermore,
$$ 
F_k(t,T_1,T_2) = S_k(t) + \sum_{n=-k}^k G_n(t,T_1,T_2) \left(e^{\lambda_n t} \<g_n^*,f_k(0)\>_{\alpha} 
+ \int_0^t e^{\lambda_n (t-s)}dX_n(s)\right)\,, 
$$
  for any $t\leq T_1\leq T_2\leq T$ where $S_k(t) = \delta_0(f_k(t))$, 
  $$G_n(t,T_1,T_2) = \frac{\exp(\lambda_n(T_2-t))-\exp(\lambda_n(T_1-t))-\lambda_n(T_2-T_1)}{\lambda_n^2\sqrt{T}(T_2-T_1)}$$
  and $X_n(t):=\int_0^t\<\Pi\beta(s)ds + \Pi\Psi(s)dL(s),g_n^*\>_{\alpha}$.
\end{cor}
\begin{proof}
  Theorem \ref{t:main statement} yields uniform $L^2$ convergence of the integrands appearing in $F_k$ to the integrand appearing in $F$ and hence the convergence follows. The representation of $F_k$ follows immediately from part (3) of Theorem~\ref{t:main statement}.
\end{proof}
We remark in passing that temperature derivatives market (see e.g. Benth and \v{S}altyt\.{e} Benth~\cite{BSB-weather}) trades in forwards with a "delivery period" as well. 
In this market, the forward is cash-settled against an index of the daily average temperature measured in a city over a given period.

\section{Refinement to Markovian forward price models}\label{s:markovian}
In this Section we refine our analysis to Markovian forward price models, making the
additional assumption that the coefficients $\beta$ and $\Psi$ depend on the state of the forward curve.
More specifically, we assume that
 \begin{align}
    \beta(t) &= b(t,f(t)), \\
    \Psi(t) &= \psi(t,f(t)),
 \end{align}
where $b:\mathbb R_+\times H_\alpha\rightarrow H_\alpha$, $\psi:\mathbb R_+\times H_\alpha\rightarrow L(H_\alpha)$ are measurable Lipschitz-continuous functions of linear growth
in the sense
 \begin{align}
     \| b(t,f) - b(t,g)\|_{\alpha} &\leq C_b \| f-g\|_{\alpha}\,, \label{eq:lip-cond-b} \\
     \Vert (\psi(t,f) - \psi(t,g))\mathcal Q^{1/2} \Vert_{\mathrm{HS}} &\leq C_\psi \| f-g\|_{\alpha}\,, 
 \label{eq:lip-cond-psi} 
  \end{align}
and
 \begin{align}
     \| b(t,f)\|_{\alpha} &\leq C_b(1+ \|f\|_{\alpha})\,, \label{eq:lingrowth-cond-b} \\
     \Vert \psi(t,f)\mathcal Q^{1/2} \Vert_{\mathrm{HS}} &\leq C_\psi (1+ \|f\|_{\alpha})\,, 
 \label{eq:lingrowth-cond-psi} 
  \end{align}
for positive constants $C_b$, $C_\psi$. Under these conditions there exists a unique mild solution $f$ of 
the semilinear SPDE
\begin{equation}
\label{eq:nonlinear_spde}
df(t) = (\partial_xf(t) + b(t,f(t))) dt + \psi(t,f(t-)) dL(t),\quad f(0) = f_0.
\end{equation}
We would like to note that semilinear SPDEs are treated in the book by Peszat and Zabczyk~\cite{peszat.zabczyk.07} and in Tappe~\cite{tappe.12}. Additionally, we assume that
  \begin{align}
    b(t,h) &= b(t,g), \label{eq:struct-cond-b}\\
    \psi(t,h) &= \psi(t,g) \label{eq:struct-cond-psi}\,,
  \end{align}
 for any $h,g\in H_\alpha$ such that $h(x)=g(x)$ for any $x\in[0,T-t]$, i.e.\ the structure of the curve beyond our time horizon $T$ does not influence the 
dynamics of the curve-valued process $f(t)$.

Before continuing our analysis of the arbitrage-free approximation in the Markovian case, we show a couple of useful lemmas. The first states a
version of Doob's $L^2$ inequality for Volterra-like Hilbert space-valued stochastic integrals with respect to the L\'evy process $L$, and is essentially collected from 
Filipovi\'c, Tappe and Teichmann~\cite{FTT}.
\begin{lem}
\label{lem:doob}
Suppose that 
$\Phi\in\mathcal L_L^2(H_{\alpha})$. Then,
\begin{align*} 
\E\left[\sup_{s\in[0,t]} \Vert
\int_0^s\mathcal{U}_{s-r}\Phi(r)\,dL(r)
\Vert_{\alpha}^2\right] &\leq 4c_t^2 \int_0^t \E\left[\Vert\Phi(r)\mathcal Q^{1/2}\Vert_{\text{HS}}^2\right]\,dr\,,
\end{align*}
for $c_t>0$ being at most exponentially growing in $t$. 
\end{lem}
\begin{proof}
Note first that due to Benth and Kr\"uhner~\cite[Lemma 3.5]{benth.kruehner.14} the $C_0$-semigroup $(\mathcal U_t)_{t\geq 0}$ is pseudo-contractive. Filipovi\'c, Tappe and Teichmann~\cite[Prop.~8.7]{FTT} state that there is a Hilbert space extension $H$
of $H_\alpha$ (i.e.\ $H$ is a Hilbert space and $H_\alpha$ is its
subspace and the norm of $H_\alpha$ equals the norm of $H$ restricted to
$H_\alpha$) and a $C_0$-group $(\mathcal V_t)_{t\in\mathbb R}$ on $H$ such that
$\mathcal V_t\vert_{H_\alpha} = \mathcal U_t$ for $t\geq 0$. Then, we have
\begin{align*}
\sup_{s\in[0,t]} \Vert\int_0^s\mathcal{U}_{s-r}\Phi(r)\,dL(r)\Vert_{\alpha}&\leq \sup_{s\in[0,t]}\Vert \mathcal V_{s-t}\Vert_{\mathrm{op}}
\Vert\int_0^s \mathcal U_{t-r}\Phi(r)\,dL(r) \Vert_{\alpha} \\
&\leq \sup_{s\in[0,t]}\Vert \mathcal V_{s}\Vert_{\mathrm{op}}
\sup_{s\in[0,t]}\Vert\int_0^s \mathcal U_{t-r}\Phi(r)\,dL(r) \Vert_{\alpha}\,.
\end{align*}
Thus, by Doob's maximal inequality, Thm.~2.2.7 in 
Prevot and R\"ockner~\cite{PR}, we find
\begin{align*} 
\E&\left[\sup_{s\in[0,t]} \Vert
\int_0^s\mathcal{U}_{s-r}\Phi(r)\,dL(r)
\Vert_{\alpha}^2\right]   \\
&\qquad\qquad\leq \sup_{s\in[0,t]}\Vert \mathcal
V_{s}\Vert_{\mathrm{op}}^2\E\left[\sup_{s\in[0,t]} \Vert
\int_0^s\mathcal{U}_{t-r}\Phi(r)\,dL(r)\Vert_{\alpha}^2\right]  \\
&\qquad\qquad\leq 4\sup_{s\in[0,t]}\Vert \mathcal
V_{s}\Vert_{\mathrm{op}}^2\E\left[\Vert\int_0^t\mathcal{U}_{t-r}\Phi(r)\,dL(r)\Vert_{\alpha}^2\right]  \\
&\qquad\qquad=4\sup_{s\in[0,t]}\Vert \mathcal
V_{s}\Vert_{\mathrm{op}}^2\int_0^t\E\left[\|\mathcal U_{t-r}\Phi(r)\mathcal{Q}^{1/2}\Vert^2_{\text{HS}}\right]\,dr \\
&\qquad\qquad\leq 4\sup_{s\in[0,t]}\Vert \mathcal
V_{s}\Vert_{\mathrm{op}}^2\sup_{s\in[0,t]}\|\mathcal U_{s}\Vert_{\text{op}}^2\int_0^t\E\left[ \Vert\Phi(r)\mathcal{Q}^{1/2}\Vert^2_{\text{HS}}\right]\,dr 
\end{align*}
This proves the Lemma by letting $c_t=\sup_{s\in[0,t]}\Vert \mathcal
V_{s}\Vert_{\mathrm{op}}\sup_{0\leq s\leq t}\Vert\mathcal{U}_{s}\Vert_{\text{op}}$
and recalling that any $C_0$-group is bounded in operator norm by
an exponentially increasing function in $t$. Hence, $c_t\leq c\exp(w t)$ for some
constants $c,w>0$. 
\end{proof}
We remark in passing that the above result holds for any pseudo-contractive semigroup 
$\mathcal S_t$, $t\geq 0$. 
 
The next lemma is a useful technical result on the distance between processes and the fixed point of an integral operator defined
via the mild solution of \eqref{eq:nonlinear_spde}. The lemma plays a crucial role in showing that certain arbitrage-free approximations of 
\eqref{eq:nonlinear_spde} converge to the right limit. 
\begin{lem}\label{l:fixpoint estimate}
For an $H_{\alpha}$-valued adapted and c\`adl\`ag stochastic process $h$, define
$$ 
V(h)(t) := \mathcal U_tf_0 + \int_0^t \mathcal U_{t-s}b(s,h(s))\,ds + \int_0^t \mathcal U_{t-s}\psi(s,h(s-))\,dL(s)\,,
$$
for any $t\geq 0$. Then, $V$ has a fixed point $\widehat{f}$ and it holds
$$ 
\E\left[\sup_{0\leq s\leq t}\| h(s)-\widehat{f}(s)\|^2_{\alpha} \right] \leq\frac{\pi^2}{6}\exp(4C_t)\E\left[\sup_{0\leq s\leq t}\| V(h)(s)-h(s)\|^2_{\alpha}\right] \,,
$$
for any $t\geq 0$ and any $H_{\alpha}$-valued adapted c\`adl\`ag stochastic processes $h$, with
$C_t$ being a positive constant depending on $t$. 
\end{lem}
\begin{proof}
If $h$ is an adapted c\`adl\`ag $H_{\alpha}$-valued 
stochastic process such that $\mathbb E[\int_0^t\|h(s)\|_{\alpha}^2\,ds]<\infty$, then from the linear growth assumption~\eqref{eq:lingrowth-cond-b}
on $b$ we find 
\begin{align*}
\mathbb E[\int_0^t\|\mathcal U_{t-s}b(s,h(s))\|_{\alpha}\,ds] &\leq C_b e^{wt}(t+\mathbb E[\int_0^t\|h(s)\|_{\alpha}\,ds]) \\
&\leq C_b e^{wt}(t+\sqrt{t}\mathbb E[\int_0^t\|h(s)\|_{\alpha}^2\,ds]^{1/2})\\
&<\infty\,.
\end{align*}
Furthermore, from the linear growth condition~\eqref{eq:lingrowth-cond-psi} 
on $\psi$
$$
\mathbb E[\int_0^t\|\mathcal{U}_{t-s}\psi(s,h(s))\|^2_{\alpha}\,ds]\leq 2C^2_{\psi}e^{2wt}\left(t+\mathbb E[\int_0^t\|h(s)\|^2_{\alpha}\,ds]\right)<\infty\,.
$$
Hence, $V(h)$ is well-defined, and it is an adapted c\`adl\`ag process. By a straightforward 
estimation using again the linear growth of $b$ and $\psi$, we find similarly that
$$
\mathbb E[\int_0^t\|V(h)(s)\|_{\alpha}^2\,ds]\leq C_t\left(1+\mathbb E[\int_0^t\|h\|_{\alpha}^2\,ds]\right)<\infty\,,
$$
for some constant $C_t>0$
Therefore, $V$ maps into its own domain and, thus, can be iterated.

We note that by general theory, the SPDE
$$
df(t)=\partial_xf(t)\,dt+b(t,f(t))\,dt+\psi(t,f(t-))\,dL(t)
$$
has a unique mild solution $\widehat f$ which has a c\`adl\`ag modification, cf.\ Tappe~\cite[Theorem 4.5, Remark 4.6]{tappe.12}. By definition of mild solutions, we see that $\widehat f$ 
is a fix point for
$V$, i.e., $V(\widehat f)=\widehat f$. 

Let $g,h$ be $H_{\alpha}$-valued adapted c\`adl\`ag stochastic processes and $t\geq 0$. Then, we have
\begin{align*}
\E&\left[\sup_{0\leq s\leq t}\|V(h)(s)-V(g)(s)\|^2_{\alpha}\right] \\
&\qquad\qquad\leq 2\E\left[\sup_{0\leq s\leq t}\|\int_0^s\mathcal U_{s-r}\left(b(r,h(r))-b(r,g(r))\right)\,dr\|_{\alpha}^2\right] \\
&\qquad\qquad\qquad+2\E\left[\sup_{0\leq s\leq t}\|\int_0^s\mathcal U_{s-r}\left(\psi(r,h(r-))-\psi(r,g(r-))\right)\,dL(r)\|^2_{\alpha}\right] \,.
\end{align*}
Consider the first term on the right hand side of the inequality. By the norm inequality for Bochner integrals and Lipschitz continuity of $b$ in 
\eqref{eq:lip-cond-b}, we find
\begin{align*}
\E&\left[\sup_{0\leq s\leq t}\|\int_0^s\mathcal U_{s-r}\left(b(r,h(r))-b(r,g(r))\right)\,dr\|_{\alpha}^2\right]  \\
&\qquad\qquad\leq\E\left[\sup_{0\leq s\leq t}\left(\int_0^s\|\mathcal U_{s-r}\|_{\text{op}}\|b(r,h(r))-b(r,g(r))\|_{\alpha}\,dr\right)^2\right] \\
&\qquad\qquad\leq t\E\left[\sup_{0\leq s\leq t}\int_0^s\|\mathcal U_{s-r}\|^2_{\text{op}}\|b(r,h(r))-b(r,g(r))\|^2_{\alpha}\,dr\right] \\
&\qquad\qquad\leq t^2\sup_{0\leq s\leq t}\|\mathcal{U}_s\|^2_{\text{op}}\E\left[\int_0^t\|b(r,h(r))-b(r,g(r))\|_{\alpha}^2\,dr\right] \\
&\qquad\qquad\leq t^2C_b^2\sup_{0\leq s\leq t}\|\mathcal U_s\|^2_{\text{op}}\int_0^t\E\left[\|h(r)-g(r)\|_{\alpha}^2\right]\,dr\,,
\end{align*} 
where we have applied Cauchy-Schwartz' inequality. Recall that since
$\mathcal U_t$ is a pseudo-contractive semigroup, we find for some
$w>0$, it holds that $\sup_{0\leq s\leq t}\|\mathcal U_{s}\|_{\text{op}}^2\leq \exp(2w t)<\infty$.

For the second term, we find by appealing to Lemma~\ref{lem:doob} and the Lipschitz continuity in \eqref{eq:lip-cond-psi} of $\psi$,
\begin{align*}
\E&\left[\sup_{0\leq s\leq t}\|\int_0^s\mathcal U_{s-r}\left(\psi(r,h(r-))-\psi(r,g(r-))\right)\,dL(r)\|^2_{\alpha}\right] \\
&\qquad\qquad\leq 4c_t^2\int_0^t\E\left[\Vert(\psi(r,h(r))-\psi(r,g(r)))\mathcal Q^{1/2}\Vert^2_{\text{HS}}\right]\,dr \\
&\qquad\qquad\leq 4c_t^2C_{\psi}^2\int_0^t\E\left[\Vert h(r)-g(r)\Vert^2_{\alpha}\right]\,dr 
\end{align*}
Here, the constant $c_t$ is from Lemma~\ref{lem:doob}. Denote by $C_t$ the constant
$$
C_t:=2C^2_b t^2 \sup_{s\in[0,t]}\Vert \mathcal U_s\Vert_{\mathrm{op}}+ 8c^2_tC_\psi^2t\,.
$$
Then, we have
\begin{align*}
  \E&\left[\sup_{0\leq s\leq t}\| V^n(h)(s)-V^n(g)(s)\|^2_{\alpha}\right] \\ 
&\qquad\qquad\leq C_t\int_0^t \E\left[\|V^{n-1}(h)(s_1)-V^{n-1}(g)(s_1)\|^2_{\alpha}\right] \,ds_1\\
   &\qquad\qquad\leq C_t^n \int_0^t\int_0^{s_1}\cdots\int_0^{s_{n-1}}\E\left[\|h(s_n)-g(s_n)\|_{\alpha}^2\right]ds_n\dots ds_1 \\
   &\qquad\qquad\leq \frac{C_t^n}{n!}\E\left[\sup_{0\leq s\leq t}\|h(s)-g(s)\|^2_{\alpha} \right]\,,
\end{align*}
for any $n\in\mathbb N$. Denote by $L_{a}^2(\Omega,D([0,t],H_\alpha))$ the space of 
$H_{\alpha}$-valued adapted c\`adl\`ag stochastic processes 
$\{f(s)\}_{s\in[0,t]}$ for which $\E[\sup_{s\in[0,t]}\|f(s)\|_\alpha^2]<\infty$.
Equip this space with the norm $\|\cdot\|_t$ defined by
$$ 
\|f\|_t^2 := \E[\sup_{s\in[0,t]}\|f(s)\|_\alpha^2]
$$
for $f\in L_{a}^2(\Omega,D([0,t],H_\alpha))$. From the estimation above, we see
that $V$ operates on the normed space $L_{a}^2(\Omega,D([0,t],H_\alpha))$.
Moreover, $V^n$ is Lipschitz continuous with constant strictly less than $1$ for $n$ sufficiently large. Thus, by Banach's fixed point theorem there is at most one 
fixed point for $V$. Hence, $\hat f$ is the unique fix point for $V$. Furthermore, we have
\begin{align*}
\E\left[\sup_{0\leq s\leq t}\|V^n(h)(s)-h(s)\|^2_{\alpha}\right]^{1/2} 
& \leq \sum_{k=0}^{n-1} \E\left[\sup_{0\leq s\leq t}\|V^{k+1}(h)(s)-V^k(h)(s)\|^2_{\alpha}\right]^{1/2} \\
                                        & \leq \E\left[\sup_{0\leq s\leq t}\| V(h)(s)-h(s)\|^2_{\alpha} \right]^{1/2}\sum_{k=0}^{n-1} \left(\frac{C_t^k}{k!}\right)^{1/2}\,.
\end{align*}
From Cauchy-Schwartz' inequality and we have that
\begin{align*}
\sum_{k=0}^{n-1} \left(\frac{C_t^k}{k!}\right)^{1/2}&=\sum_{k=0}^{n-1}(k+1)^{-1}\left(\frac{(k+1)^2C_t^k}{k!}\right)^{1/2} \\
&\leq\left(\sum_{k=0}^{n-1}\frac1{(k+1)^2}\right)^{1/2}\left(\sum_{k=0}^{n-1}\frac{(k+1)^2C_t^k}{k!}\right)^{1/2} \\
&\leq \frac{\pi}{\sqrt{6}}\left(\sum_{k=0}^{n-1}\frac{4^kC_t^k}{k!}\right)^{1/2} \\
&\leq \frac{\pi}{\sqrt{6}}\exp(2C_t)\,, 
\end{align*}
where we have used the elementary inequality $k+1\leq 2^k$, $k\in\mathbb N$. 
\end{proof}

Let us define the Lipschitz continuous functions $b_\Pi:=\Pi\circ b$ and $\psi_\Pi:=\Pi\circ\psi$. Then, 
Tappe~\cite[Theorem 4.5]{tappe.12} yields a mild solution $f_\Pi$ for the SPDE
\begin{equation}
df_\Pi(t) = (\partial_xf_\Pi(t) + b_\Pi(t,f_\Pi(t)))\, dt + \psi_\Pi(t,f_\Pi(t-))\,dL(t),\quad f_\Pi(0) = \Pi f_0\,. 
\end{equation}
 Furthermore, it will be convenient to use the notations
  \begin{align}
    b_k(t,h) := \Lambda_k(b(t,h)), \\
    \psi_k(t,h) := \Lambda_k(\psi(t,h))
  \end{align}
 for any $h\in H_\alpha$, $t\geq0$.

In the proof of Theorem~\ref{t:main statement} we compared the solution $f$ to the projected solution $\Pi f$ which are essentially the same due to properties of $\Pi$. Then we compared $\Pi f$ to $f_\Pi$ which again had been essentially the same. Finally, we compared $\Pi_kf_\Pi$ to solutions of the projected SPDE where the difference was given by a certain Lie-commutator. 
However, in the Markovian setting we want to change the dependencies of the coefficients as well, which complicates the proof of the approximation result.
\begin{thm}
\label{thm:main-markovian}
Denote by $\widehat f_k$ be the mild solution of the SPDE
$$ 
d\widehat f_k(t) = (\partial_x\widehat f_k(t) + b_k(t,\widehat f_k(t)))\,dt + \psi_k(t,\widehat f_k(t-))\,dL(t),\quad \widehat f_k(0) = \Lambda_k f_0, t\geq0\,.
$$
Then, $\widehat f_k\in H_{\alpha}^{T,k}$ is a strong solution, and we have
  $$ \E\left[\sup_{t\in[0,T],x\in[0,T-t]} \vert \hat f_k(t,x)-f(t,x) \vert^2 \right] \rightarrow 0 $$
  for $k\rightarrow \infty$.
\end{thm}
\begin{proof}
First we note that a unique mild solution $\widehat{f}_k$ of the SPDE exists due to Tappe~\cite[Theorem 4.5]{tappe.12}.
Define
$$ 
V_k(h)(t) := \mathcal U_t f_k(0) + \int_0^t \mathcal U_{t-s} (b_k(s,h(s))\,ds + \psi_k(s,h(s-))\,dL(s))\,, 
$$
for any $k\in\mathbb N$, $t\geq 0$ and any adapted c\`adl\`ag stochastic process $h$ in $H_{\alpha}$. 
Let $f_k$ be
defined as
\begin{align*}
    f_k(t): &= \mathcal U_t f_k(0) + \int_0^t \mathcal U_{t-s} (b_k(s,f(s))\,ds + \psi_k(s,f(s))\,dL(s) \\
           &= \mathcal U_t f_k(0) + \int_0^t \mathcal U_{t-s} (b_k(s,f_\Pi(s))\,ds + \psi_k(s,f_\Pi(s-))\,dL(s) \\
           &= V_k(f_\Pi)(t)\,,
\end{align*}
for $f_k(0)=\Lambda_kf(0)$. Moreover, $\widehat f_k(t) = V_k(\widehat f_k)(t)$. By Lemma~\ref{l:fixpoint estimate}, it holds
  \begin{align*}
\E\left[\sup_{0\leq s\leq t}\| f_\Pi(t) - \hat f_k(t) \|_\alpha^2\right]\leq\frac{\pi^2}{6}\exp(4C_t) \E\left[ \sup_{0\leq s\leq t}\| f_k(s)-f_\Pi(s)\|^ 2_\alpha\right]\,,
  \end{align*}
  for any $k\in \mathbb N$, $t\geq 0$ and $C_t$ given in the lemma (recall from Section~\ref{s:the model} that the operator norm of
the shift semigroup $\mathcal{U}_t$ is uniformly bounded by the constant $C_{\mathcal{U}}$).  
By the definition of $f_k$ and $f_{\Pi}$ we find
\begin{align*}
\|f_k(s)-f_{\Pi}(s)\|_{\alpha}^2&\leq2\|\int_0^s\U_{s-r}(b_k(r,f_{\Pi}(r))-b_{\Pi}(r,f_{\Pi}(r)))\,dr\|_{\alpha}^2 \\
&\qquad+2\|\int_0^s\U_{s-r}(\psi_k(r,f_{\Pi}(r-))-\psi_{\Pi}(r,f_{\Pi}(r-)))\,dL(r)\|_{\alpha}^2\,.
\end{align*}
Consider the first term on the right-hand side of the inequality. By the norm inequality for Bochner integrals, Cauchy-Schwartz' inequality and 
boundedness of the operator norm of
$\U_t$ we find (for $s\leq t$)
\begin{align*}
\|\int_0^s\U_{s-r}&(b_k(r,f_{\Pi}(r))-b_{\Pi}(r,f_{\Pi}(r)))\,dr\|_{\alpha}^2 \\
&\qquad\qquad\leq\left(\int_0^s\|\U_{s-r}(b_k(r,f_{\Pi}(r))-b_{\Pi}(r,f_{\Pi}(r)))\|_{\alpha}\,dr\right)^2 \\
&\qquad\qquad\leq t\int_0^t\|\U_{s-r}(b_k(r,f_{\Pi}(r))-b_{\Pi}(r,f_{\Pi}(r)))\|^2_{\alpha}\,dr \\
&\qquad\qquad\leq tC^2_{\U}\int_0^t\|b_k(r,f_{\Pi}(r))-b_{\Pi}(r,f_{\Pi}(r))\|_{\alpha}^2\,dr \\
&\qquad\qquad\leq tC_{\U}^2\int_0^t\|(\Pi_k-\mathcal I)b_{\Pi}(r,f_{\Pi}(r))\|_{\alpha}^2\,dr
\end{align*}
Here, $\mathcal I$ denotes the identity operator on $H_\alpha^T$. Hence, using
Lemma~\ref{lem:doob} and the fact that $\{\mathcal U\}_{t\geq 0}$ is pseudo-contractive,
  \begin{align*}
    \E&\left[ \sup_{0\leq s\leq t}\Vert f_k(s)-f_\Pi(s)\Vert^ 2_\alpha\right]  \\
&\qquad\leq 2  tC_{\U}^2\int_0^t\E\left[\|(\Pi_k-\mathcal I)b_{\Pi}(r,f_{\Pi}(r))\|_{\alpha}^2\right]\,dr \\ 
    &\qquad\qquad + 2\E\left[\sup_{0\leq s\leq t}\Vert\int_0^s \U_{s-r}(\psi_k(r,f_\Pi(r-))-\psi_\Pi(r,f_\Pi(r-)))\,dL(r)\|_{\alpha}^2\right] \\
&\qquad\leq 2  tC_{\U}^2\int_0^t\E\left[\|(\Pi_k-\mathcal I)b_{\Pi}(r,f_{\Pi}(r))\|_{\alpha}^2\right]\,dr \\ 
    &\qquad\qquad + 8c_t^2\int_0^t\E\left[\|(\psi_k(r,f_\Pi(r))-\psi_\Pi(r,f_\Pi(r)))\mathcal Q^{1/2}\|_{\text{HS}}^2\right]\,dr \\
&\qquad\leq 2  tC_{\U}^2\int_0^t\E\left[\|(\Pi_k-\mathcal I)b_{\Pi}(r,f_{\Pi}(r))\|_{\alpha}^2\right]\,dr \\ 
    &\qquad\qquad + 8c_t^2\int_0^t\E\left[\|(\Pi_k-\mathcal I)\psi_\Pi(r,f_\Pi(r))\mathcal Q^{1/2}\|_{\text{HS}}^2\right]\,dr \,.
\end{align*}
Denote by
\begin{align*}
K_t(k):&=2  tC_{\U}^2\int_0^t\E\left[\|(\Pi_k-\mathcal I)b_{\Pi}(r,f_{\Pi}(r))\|_{\alpha}^2\right]\,dr \\
&\qquad+8c_t^2\int_0^t\E\left[\|(\Pi_k-\mathcal I)\psi_\Pi(r,f_\Pi(r))\mathcal Q^{1/2}\|_{\text{HS}}^2\right]\,dr \,,
\end{align*}
for $k\in\mathbb N$. By standard norm inequalities, we have
\begin{align*}
K_t(k):&=4  tC_{\U}^2(1+\|\Pi_k\|^2_{\text{op}})\int_0^t\E\left[\|b_{\Pi}(r,f_{\Pi}(r))\|_{\alpha}^2\right]\,dr \\
&\qquad+16c_t^2(1+\|\Pi_k\|^2_{\text{op}})\int_0^t\E\left[\|\psi_\Pi(r,f_\Pi(r))\|_{\text{op}}^2\right]\,dr \,,
\end{align*}
which is seen to be bounded uniformly in $k\in\mathbb N$ from 
Proposition~\ref{l:commutator of U and projectors}. 
Hence, we have $ K_{t}(k) \rightarrow 0$ for $k\rightarrow \infty$ and any $t\geq 0$ by the dominated convergence theorem because $(\Pi_k-\mathcal I)h\rightarrow 0$ for $k\rightarrow \infty$ and any $h\in H_\alpha^T$.
Thus, we find 
$$
\E\left[\sup_{0\leq t\leq T}\Vert f_k(t) - \hat f_k(t) \Vert_\alpha^2\right] \rightarrow 0\,,
$$ 
for $k\rightarrow \infty$. Finally, $f_\Pi(t,x) = f(t,x)$ for any $t\in[0,T]$, $x\in[0,T-t]$. Moreover,
from Lemma~3.2 in Benth and Kr\"uhner~\cite{benth.kruehner.14} the sup-norm is dominated by the 
$H_{\alpha}$-norm, and therefore we have
 \begin{align*}
   \E\left[\sup_{t\in[0,T],x\in[T-t]}\vert \hat f_k(t,x)- f(t,x)\vert^2\right] &\leq c \E\left[\sup_{0\leq t\leq T}\Vert \hat f_k(t)-f_\Pi(t) \Vert_\alpha^2\right]\rightarrow0\,,
 \end{align*}
 for $k\rightarrow \infty$.
The Proposition follows.
\end{proof}
The philosophy in Thm.~\ref{thm:main-markovian} is to take $f(t)$ as the actual forward curve dynamics, and study 
finite dimensional approximations $\widehat f_k(t)$ of it. By construction, $\widehat f_k$ solves a HJMM dynamics which yields
that the approximating forward curves become arbitrage-free. From the main theorem, the approximations
$\widehat f_k(t)$ converge uniformly to $f(t)$ for $x\in[0,T-t]$. As time $t$ progresses, the times to maturity $x\geq 0$ 
for which we obtain convergence shrink. The reason is that information of $f$ is transported 
to the left in the dynamics of the SPDE. We recall that the approximation of $f$ is constructed by first 
localizing $f$ to $x\in[0,T]$ for a fixed time horizon $T$ by the projection operator $\Pi$ down
to $H_{\alpha}^T$, and next creating finite-dimensional approximations of this.

Alternatively, we may use $f_{\Pi}(t)$ as our forward price model. Then, the finite dimensional 
approximation $f_k(t)$ will converge uniformly over all $x\in[0,T]$. In practice, there will be a time
horizon for the futures market for which we have no information. For example, in liberalized power markets
like NordPool and EEX, there are no futures contracts traded with settlement beyond 6 years. Hence,
it is a delicate task to model the dynamics of the futures price curve beyond this horizon. The alternative is
then clearly to restrict the modelling perspective to the dynamics with the maturities confined in 
$x\in[0,T]$.  
Indeed, in such a context the
structural conditions \eqref{eq:struct-cond-b} and \eqref{eq:struct-cond-psi} will be trivially satisfied as
we restrict our model parameters in any case to the behaviour on $x\in[0,T]$. 

We end our paper with a short discussion on a possible numerical implementation of $\widehat f_k(t)$, the 
finite-dimensional approximation of $f(t)$. Since $\widehat f_k(t)\in H_{\alpha}^{T,k}$, we can express it as
$$
\widehat f_k(t)=\widehat f_{k,*}(t)+\sum_{n=-k}^kg_n\widehat f_{k,n}(t)\,,
$$
where $\widehat f_{k,*}(t)=\widehat f_k(t,0)g_*$ and $\widehat f_{k,n}(t)=\<\widehat f_k(t),g_n^*\>_{\alpha}$ are $\mathbb C$-valued functions. For any $h\in H_{\alpha}^{T,k}$ it follows that $b_k(t,h)\in H_{\alpha}^{T,k}$. Define for $n=-k,\ldots,k$ the functions
\begin{align*}
\overline{b}_{k,n}&:\mathbb R_+\times \mathbb C^{2k+2}\rightarrow \mathbb C\,; \qquad (t,x_*,x_{-k},\ldots,x_k) \mapsto \left\<b_k(t,x_*g_*+\sum_{j=-k}^kx_jg_j),g_n^*\right\>_{\alpha}\,, \\
\overline{b}_{k,*}&:\mathbb R_+\times \mathbb C^{2k+2}\rightarrow \mathbb C\,; \qquad (t,x_*,x_{-k},\ldots,x_k) \mapsto \left\<b_*(t,x_*g_*+\sum_{j=-k}^kx_jg_j),g_n^*\right\>_{\alpha}\,.
\end{align*}
Furthermore, $\psi_k(t,h)\in L_{\text{HS}}(H_{\alpha},H_{\alpha}^{T,k})$. Thus, for any
$g\in H_{\alpha}$ we have that $\psi_k(t,h)(g)\in H_{\alpha}^{T,k}$. We define the mappings
\begin{align*}
\overline{\psi}_{k,n}&:\mathbb R_+\times\mathbb C^{2k+2}\rightarrow H^*_{\alpha};
(t,x_*,x_{-k},\ldots,x_k)\mapsto \left\<\psi_k(t,x_*g_*+\sum_{j=-k}^kx_jg_j)(\cdot),g_n^*\right\>_{\alpha} \\
\overline{\psi}_{k,*}&:\mathbb R_+\times\mathbb C^{2k+2}\rightarrow H^*_{\alpha};
(t,x_*,x_{-k},\ldots,x_k)\mapsto \left\<\psi_*(t,x_*g_*+\sum_{j=-k}^kx_jg_j)(\cdot),g_n^*\right\>_{\alpha}
\end{align*}
for $n=-k,\ldots,k$. Now, since $\partial_x g_*=0$ and $\partial_xg_n=\lambda_ng_n+g_*/\sqrt{T}$, we find 
from the SPDE of $\widehat{f}_k$ the following $2k+2$ system of stochastic differential equations (after comparing terms with respect to the Riesz basis functions),
\begin{align*}
d\widehat{f}_{k,*}(t)&=\left(\frac1{\sqrt{T}}\sum_{n=-k}^k\widehat f_{k,n}(t)+
\overline{b}_{k,*}(t,\widehat{f}_{k,*}(t),\widehat{f}_{k,-k}(t),\ldots,\widehat{f}_{k,k}(t))\right)\,dt \\
&\qquad\qquad+
d\overline{\psi}_{k,*}(t,\widehat{f}_{k,*}(t-),\widehat{f}_{k,-k}(t-),\ldots,\widehat{f}_{k,k}(t-))(L(t)) \\
d\widehat{f}_{k,-k}(t)&=\left(\lambda_{-k}\widehat f_{k,-k}(t)+
\overline{b}_{k,-k}(t,\widehat{f}_{k,*}(t),\widehat{f}_{k,-k}(t),\ldots,\widehat{f}_{k,k}(t))\right)\,dt \\
&\qquad\qquad+
d\overline{\psi}_{k,-k}(t,\widehat{f}_{k,*}(t-),\widehat{f}_{k,-k}(t-),\ldots,\widehat{f}_{k,k}(t-))(L(t)) \\
\cdot & \cdots \\
\cdot & \cdots \\
d\widehat{f}_{k,k}(t)&=\left(\lambda_{k}\widehat f_{k,k}(t)+
\overline{b}_{k,k}(t,\widehat{f}_{k,*}(t),\widehat{f}_{k,-k}(t),\ldots,\widehat{f}_{k,k}(t))\right)\,dt \\
&\qquad\qquad+
d\overline{\psi}_{k,k}(t,\widehat{f}_{k,*}(t-),\widehat{f}_{k,-k}(t-),\ldots,\widehat{f}_{k,k}(t-))(L(t)) 
\end{align*}
In a compact matrix notation, defining $\mathbf{x}(t)=(x_1(t),x_2(t),\ldots,x_{2k+2}(t))'$ and 
$$
A=\left[\begin{array}{ccccc} \frac1{\sqrt{T}} & \frac1{\sqrt{T}} & \frac1{\sqrt{T}} &\cdots & \frac1{\sqrt{T}} \\
0 & \lambda_{-k} & 0 & \cdots & 0 \\
0 & 0 & \lambda_{-k+1} & \cdots & 0 \\
\cdot & \cdot & \cdot & \cdots & \cdot \\
\cdot & \cdot & \cdot & \cdots & \cdot \\
0 & 0 & 0 & \cdots & \lambda_k \end{array}
\right]\,,
$$
we have the dynamics
$$
d\mathbf{x}(t)=(A\mathbf{x}(t)+\overline{\mathbf{b}}_k(t,\mathbf{x}(t)))\,dt+d\overline{\mathbf{\psi}}_k(t,\mathbf{x}(t-))(L(t))\,,
$$
with $\widehat{f}_{k,*}=x_1, \widehat{f}_{k,-k}=x_2,\ldots,\widehat{f}_{k,k}=x_k$.
Using for example an Euler approximation, we can derive an iterative numerical scheme for this 
stochastic differential
equation in $\mathbb C^{2k+2}$. We refer to Kloeden and Platen~\cite{KP} for a detailed analysis
of numerical solution of stochastic differential equations driven by Wiener noise.

\end{document}